%% file: main.tex
\documentclass[11pt,letterpaper]{article}

\usepackage{sort_4}

\newcommand{\rnote}[1]{}
\newcommand{\snote}[1]{}
\newcommand{\pnote}[1]{}

\begin{document}

\title{The SDP value for random two-eigenvalue CSPs}
\author{Sidhanth Mohanty\thanks{EECS Department, University of California Berkeley.  Supported by NSF grant CCF-1718695} \and Ryan O'Donnell\thanks{Computer Science Department, Carnegie Mellon University.  Supported by NSF grant CCF-1717606. This material is based upon work supported by the National Science Foundation under grant numbers listed above. Any opinions, findings and conclusions or recommendations expressed in this material are those of the author and do not necessarily reflect the views of the National Science Foundation (NSF).}
 \and Pedro Paredes${}^{\text{\textdagger}}$}
\date{\today}
\maketitle

\begin{abstract}
    We precisely determine the SDP value (equivalently, quantum value) of large random instances of certain kinds of constraint satisfaction problems, ``two-eigenvalue 2CSPs''.  We show this SDP value coincides with the spectral relaxation value, possibly indicating a computational threshold.   Our analysis extends the previously resolved cases of random regular $\twoxor$ and $\naethree$, and includes new cases such as random $\sort_4$ (equivalently, $\chsh$) and $\forrelation$ CSPs.  Our techniques include new generalizations of the nonbacktracking operator, the Ihara--Bass Formula, and the Friedman/Bordenave proof of Alon's Conjecture.
\end{abstract}

\thispagestyle{empty}
\setcounter{page}{0}
\newpage
\thispagestyle{empty}
\tableofcontents
\setcounter{page}{0}
\newpage

\input{intro}

\input{prelims}

\input{ihara}

\input{bound}

\input{gaussian-wave}

\input{bordenave2}

\input{SDP-value}

\section*{Acknowledgments}  \label{sec:ack}
We thank Yuval Peled for emphasizing the bipartite graph view of additive lifts, and Tselil Schramm for helpful discussions surrounding the trace method on graphs.  S.M.\ would like to thank Jess Banks and Prasad Raghavendra for plenty of helpful discussions on orthogonal polynomials and nonbacktracking walks.  Finally, we are grateful to Xinyu Wu for bringing the relevance of~\cite{bordenave2018eigenvalues} to our attention and helping us to understand the issues discussed in \pref{sec:bc}.

\bibliographystyle{alpha}
\bibliography{sort_4}

\end{document}

%% file: intro.tex
\section{Introduction}      \label{sec:intro}
This work is concerned with the average-case complexity of constraint satisfaction problems (CSPs).  In the theory of algorithms and complexity, the most difficult instances of a given CSP are arguably random (sparse) instances. Indeed, the assumed intractability of random CSPs underlies various cryptographic proposals for one-way functions~\cite{Gol00,JP00}, pseudorandom generators~\cite{BFKL94}, public key encryption~\cite{ABW10}, and indistinguishability obfuscation~\cite{Lin17}, as well as hardness results for learning~\cite{DS16} and optimization~\cite{Fei02}.  Random CSPs also provide a rich testbed for algorithmic and lower-bound techniques based on statistical physics~\cite{MM09} and convex relaxation hierarchies~\cite{KMOW17,RRS17}.

For a random, say, $\maxcut$ instance average degree~$d$, its optimum value is with high probability~(whp) concentrated around a certain function of~$d$.  Similarly, given a random $\threesat$ instance where each variable participates in an average of~$d$ clauses, the satisfiability status is whp determined by~$d$.  However explicitly working out the optimum/satisfiability as a function of~$d$ is usually enormously difficult; see, for example, Ding--Sly--Sun's landmark verification~\cite{DSS15} of the $\ksat$ threshold for sufficiently large~$k$, or Talagrand's proof~\cite{Tal06} of the Parisi formula for the Sherrington--Kirkpatrick model ($\maxcut$ with random Gaussian edge weights).  The latter was consequently used by Dembo--Montanari--Sen~\cite{DMS17} (see also~\cite{Sen18}) to determine that the $\maxcut$ value in a random $d$-regular graph is a $\frac12 + \frac{P^*}{\sqrt{d}}(1 \pm o_d(1))$ fraction of edges (whp), where $P^* \approx .7632$ is an analytic constant arising from Parisi's formula.

%Philosophically, one may believe the difficulty of these exact determinations is due to the underlying computational hardness of the CSPs themselves.  Given this, one may ask how well \emph{computationally efficient} algorithms can determine the optimum of a random sparse CSP.  Again, this question has been intensively studied from both statistical physics and convex hierarchy perspectives, and has direct bearing on the security of cryptographic schemes based on random CSPs.  XXXI want to talk about how good a solution you can find efficiently, vs what the true answer is, vs what bound you can certify efficiently\dots and possible computational thresholdsXXX.  In this paper we will consider one of the most well-studied and effective polynomial-time algorithm for CSPs: the semidefinite programming (SDP) relaxation.

\paragraph{Computational gaps for certification.}  Turning to computational issues, there are two main algorithmic tasks associated with an $n$-variable CSP: \emph{searching} for an assignment achieving large value (hopefully near to the optimum), and \emph{certifying}  (as, e.g., convex relaxations do) that no assignment achieves some larger value.  Let's take again the example of random $d$-regular $\maxcut$, where whp we have $\OPT \approx \frac12 + \frac{P^*}{\sqrt{d}}$.  It follows from~\cite{Lyo17} there is an efficient algorithm that whp finds a cut of value at least $\frac12 + \frac{2/\pi}{\sqrt{d}}$. One might say that this provides a $\frac{2}{\pi P^*}$-approximation for the search problem,\footnote{Depending on one's taste in normalization; i.e., whether one prefers the objective function $\avg_{(u,v) \in E} (\frac12 - \frac12 x_u x_v)$ or $-\avg_{(u,v) \in E} x_u x_v$, for $x \in \{\pm 1\}^V$.} where $\frac{2}{\pi P^*} \approx .83$. On the other side, the $\maxcut$ in a $d$-regular graph $G$ is always at most $\frac12 + \frac{-\lambda_{\text{min}}(G)}{2d}$, and Friedman's proof of Alon's Conjecture~\cite{Fri08} shows that $-\lambda_{\text{min}}(G) \leq 2\sqrt{d-1} + o_n(1)$ whp; thus computing the smallest eigenvalue efficiently certifies $\OPT \lessapprox \frac12 + \frac{1}{\sqrt{d}}$.  One might say that this efficient spectral algorithm provides a $\frac{1}{P^*}$-approximation for the certification problem, where $\frac{1}{P^*} \approx 1.31$.

It is a very interesting question whether either of these approximation algorithms can be improved. On one hand, it would seem desirable to have efficient algorithms that come arbitrarily close to matching the ``true'' answer on random inputs. On the other hand, the nonexistence of such algorithms would be useful for cryptography and hardness-of-approximation and -learning results.

Speaking broadly, efficient algorithms for the search problem seem to do better than efficient algorithms for the certification problem.  For example, given a random $\threesat$ instance with clause density slightly below the satisfiability threshold of $\approx 4.2667$, there are algorithms~\cite{MPR16} that seem to efficiently find satisfying assignments whp.  On the other hand, the longstanding Feige Hypothesis~\cite{Fei02} is that efficient algorithms cannot certify unsatisfiability at any large constant clause density, and indeed there is no efficient algorithm that is known to work at density~$o(\sqrt{n})$.  Similarly, for the Sherrington--Kirkpatrick model, Montanari~\cite{montanari2018optimization} has recently given an efficient PTAS for the search problem\footnote{Modulo a widely believed analytic assumption.}, whereas the best known efficient algorithm for the certification problem is again only a $1/P^*$-approximation.  These kinds of gaps seem to be closely related to ``information-computation gaps'' and Kesten--Stigum thresholds for information recovery and planted-CSP problems.

In this work we focus on potential computational thresholds for random CSP certification/refutation problems in the sparse setting, and in particular how these thresholds depend on the ``type'' of the CSP.  For CSPs with a predicate supporting a pairwise-uniform distribution --- such as $\ksat$ or $\kxor$, $k \geq 3$ --- there is solid evidence that the computational threshold for efficient certification of unsatisfiability is very far from the actual unsatisfiability threshold.  Such CSPs are whp unsatisfiable at constant constraint density, but any polynomial-time algorithm using the powerful Sum-of-Squares (SoS) algorithm fails to refute unless the density is~$\Omega(\sqrt{n/\log n})$~\cite{KMOW17}.  But outside the pairwise-supporting case, and especially for ``$\twoxor$-like'' CSPs such as $\maxcut$ and $\naethree$ (Not-All-Equal $\threesat$), the situation is much more subtle.  For one, the potential gaps are much more narrow; e.g., in random $\naethree$, even a simple spectral algorithm efficiently refutes satisfiability at constant constraint density.  Thus one must look into the actual \emph{constants} to determine if there may be an ``information-computation'' gap.  Another concern is that evidence for computational hardness in the form of SoS lower bounds (degree~$4$ or higher) seems very hard to come by (see, e.g.,~\cite{Mon17}).

\paragraph{Prior work.}  Let us describe two prior efforts towards computational thresholds for upper-bound-certification in ``$\twoxor$-like'' random CSPs.   Montanari and Sen~\cite{MS16} (see also~\cite{banks2017lovasz}) investigated the $\maxcut$ problem in random $d$-regular graphs, where the optimum value is $\frac12 + \frac{P^*}{\sqrt{d}}$ whp (ignoring $1 \pm o_d(1)$ factors).  Friedman's Theorem implies that the basic eigenvalue bound efficiently certifies the value is at most $\frac12 + \frac{1}{\sqrt{d}}$.  By using a variant of the Gaussian Wave~\cite{Elo09,CGHV15,HV15} construction for the infinite $d$-ary tree, Montanari and Sen were able to show that even the Goemans--Williamson semidefinite programming (SDP) relaxation~\cite{DP93,GW95} is still just $\frac12 + \frac{1}{\sqrt{d}}$ whp.  This may be considered evidence that \emph{no} polynomial-time algorithm can certify upper bounds better than $\frac12 + \frac{1}{\sqrt{d}}$, as Goemans--Williamson has seemed to be the optimal polynomial-time $\maxcut$ algorithm in all previous circumstances.  Of course it would be more satisfactory to see higher-degree SoS lower bounds, but as mentioned these seem very difficult to come by.

Recently, Deshpande~et~al.~\cite{deshpande2018threshold} have given similar results for random ``$c$-constraint-regular'' $\naethree$ CSPs; i.e., random instances where each variable participates in exactly~$c$ $\naethree$ constraints.\footnote{We have changed terminology to avoid a potential future confusion; we will be associating $\naethree$ constraints with triangle graphs, so $c$-constraint-regular $\naethree$ instances will be associated to $2c$-regular graphs.} Random $c$-constraint-regular instances of $\naethree$ are easily shown to be unsatisfiable (whp) for $c \geq 8$.  Deshpande~et~al.\ identified an exact threshold result for when the natural SDP algorithm is able to certify unsatisfiability: it succeeds (whp) if $c > 13.5$ and fails (whp) if $c < 13.5$.  Indeed, they show that for $c \geq 14$ even the basic spectral algorithm certifies unsatisfiability, whereas for $c \leq 13$ even the SDP augmented with ``triangle inequalities'' fails to certify unsatisfiability. Again, this gives evidence for a gap between the threshold for unsatisfiability and the threshold for computationally efficient refutation.   The techniques used by Deshpande~et~al.\ are similar to those of Montanari--Sen, except with random $(b,c)$-biregular graphs replacing random $c$-regular graphs.  (The reason is that the primal graph of a random $c$-constraint-regular $\naethree$ instance resembles the square of a random $(3,c)$-biregular graph.)

In fact, the Deshpande~et~al.\ result is more refined, being concerned not just with satisfiability of random $\naethree$ instances, but their optimal value as maximization problems.  Letting $f(c) = \frac98 - \frac38 \cdot \frac{(\sqrt{c-1} - \sqrt{2})^2}{c}$ for $c \geq 3$, they determined that in a random $c$-constraint-regular $\naethree$ instance, the SDP value is whp $f(c) \pm o(1)$; and furthermore, this is also the basic eigenvalue bound and the SDP-with-triangle-inequalities bound.  (Note that $f(13.5) = 1$.)  Again, this may suggest that in these instances, computationally efficient algorithms can only certify that at most an $f(c) + o(1)$ fraction of constraints are simultaneously satisfiable. % \rnote{actually, they show that SDP+tri thinks $f(c)-\eps$ is gettable in any nonrandom $c$-constraint-regular instance with high enough girth; do we get this too? should we mention it?}

\subsection{Our results}
The goal of the present work is to generalize the preceding Montanari--Sen and Deshpande~et~al.\ results to a broader class of sparse random 2CSPs and $\twoxor$-like optimization problems, obtaining precise values for their SDP values.  Along the way, we need to come to a deeper understanding of the combinatorial and analytic tools used (nonbacktracking walks, Ihara--Bass formulas, eigenvalues of random graphs and infinite graphs) and we need to extend these tools to graphs that do \emph{not} locally resemble trees (as in Montanari--Sen and Deshpande~et~al.).  We view this aspect of our work as a main contribution, beyond the mere statement of SDP values for specific CSPs.  We defer to \pref{sec:fried-bord} more detailed discussions of the technical conditions under which we can obtain Ihara--Bass and  Friedman-, and Gaussian Wave-type theorems.  But roughly speaking, we are able to analyze the SDP value for random regular instances of optimization problems where each ``constraint'' (not necessarily a predicate) is an \emph{edge-signed graph with two eigenvalues}.  Such constraints include:  a single edge (corresponding to random regular $\maxcut$ or $\twoxor$ as in Montanari--Sen); a complete graph (studied by Deshpande~et~al., with the $K_3$ case corresponding to random regular $\naethree$); the $\sort_4$ (a.k.a.\ $\chsh$) predicate; and,  $\forrelation_k$ constraints.  These last two have motivation from quantum mechanics, and in fact the SDP value of the associated CSPs is precisely their ``quantum value''. We discuss quantum connections further in \pref{sec:quantum}.

We state here two theorems that our new techniques allow us to prove.  Recall the $\sort_4$ predicate, which is satisfied iff its $4$ Boolean inputs $x_1, x_2, x_3, x_4$ satisfy $x_1 \leq x_2 \leq x_3 \leq x_4 \text{ or }  x_1 \geq x_2 \geq x_3 \geq x_4$.  We precisely define ``random $c$-constraint-regular CSP instance'' in \pref{sec:prelims}, but in brief, we work in the ``random lift'' model, each variable participates in exactly~$c$ constraints, and each constraint is given random negations.\footnote{Our result holds for either of the following two negation models: (i)~each \emph{constraint} is randomly negated; or, (ii)~the constraints are not negated, but each constraint is applied to random \emph{literals} rather than random variables.}
\begin{theorem}                                     \label{thm:sort}
    For random $c$-constraint-regular instances of the $\sort_4$-CSP, the SDP-satisfiability threshold occurs (in a sense) at $c = 4 + 2\sqrt{2} \approx 6.83$.  Indeed, if $c \geq 7$ then even the basic eigenvalue bound certifies unsatisfiability (whp); and, if $c \leq 6$ then the basic SDP relaxation fails to certify unsatisfiability (whp).
\end{theorem}
We remark that the trivial first-moment calculation shows that a random $c$-constraint-regular $\sort_4$-CSP is already unsatisfiable whp at degree $c = 4$.  Thus we again have evidence for a gap between the true threshold for unsatisfiability and the efficiently-certifiable threshold.

Generalizing this, the $\forrelation_k$ constraint is a certain (quantum-inspired) map $\{\pm 1\}^{2^k + 2^k} \to [-1,+1]$ that measures how correlated one $k$-bit Boolean function is with the Fourier transform of a second $k$-bit Boolean function.  We give precise details in \pref{sec:quantum}; here we just additionally remark that $\forrelation_1$ corresponds to the ``$\chsh$ game'', and that $\tfrac12 + \forrelation_1$ is equivalent to the $\sort_4$ predicate.
\begin{theorem}                                     \label{thm:forr}
    For random $c$-constraint-regular instances of the $\forrelation_k$-CSP and any constant $\eps > 0$, the SDP value is whp in the range $\frac{2\sqrt{c-1}}{c \cdot 2^{k/2}} \pm \eps$.  This is also true of the eigenvalue bound. %\rnote{and the SDP-with-triangle-inequalities value?}
\end{theorem}
When considering the SDP value for $\tfrac12 + \forrelation_1$, the formula above crosses the threshold of~$1$ when $c = 4 + 2\sqrt{2}$, yielding the statement in \pref{thm:sort} about the SDP-satisfiability threshold of random $c$-constraint-regular $\sort_4$-CSPs.

\subsection{Sketch of our techniques}
Here we sketch how our results like \pref{thm:sort} and \pref{thm:forr} are proven, using random $\sort_4$-CSPs as a running example.  A key property of the $\sort_4$ predicate is that it is essentially equivalent to the following ``$\twoxor$'' instance:

\myfig{0.15}{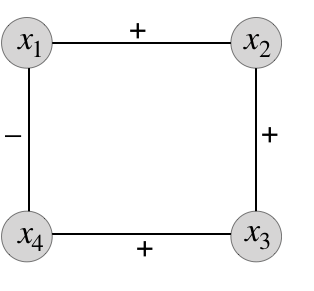}{The $\sort_4$ predicate}{fig:sort_4}

More precisely, suppose $(x_1, x_2, x_3, x_4) \in \{\pm 1\}^4$ satisfies the $\sort_4$ predicate.  Then in the graph above, exactly $3$ out of $4$ edges will be ``satisfied'' --- where an edge is considered satisfied when the product of its endpoint-labels equals the edge's label.   Conversely, if $(x_1, x_2, x_3, x_4)$ doesn't satisfy $\sort_4$ then exactly $1$ out of the $4$ edges above will be satisfied.  Now suppose we choose a random $n$-vertex $c$-constraint-regular instance $\calI$ of the $\sort_4$-CSP with, say, $c = 2$.  A small piece of such an instance might look like the following:\footnote{In fact, since we will have random negations in our instances, some $4$-cycles will have three edges labeled $-1$ and one labeled $+1$, as opposed to the other way around.  This is not an important issue for this proof sketch.}

\myfig{0.22}{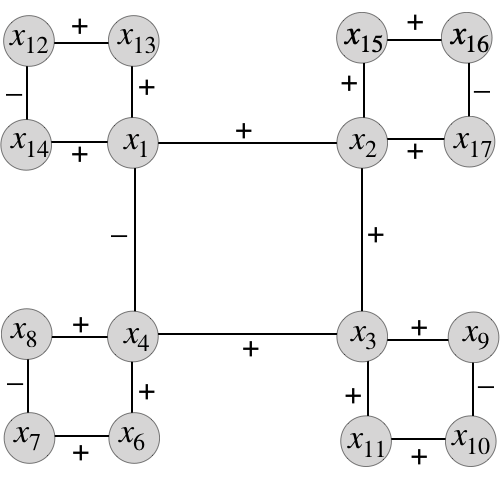}{Piece of $\sort_4$ instance}{fig:sort_4-extend}

Up to a trivial affine shift in the objective function, the optimization task is now to label the variables/vertices of $\calI$ with $\pm 1$ values $x_1, \dots, x_n$ so as to maximize $\frac{1}{n}\sum_{ij} A_{ij} x_i x_j$, where $A \in \{0, \pm 1\}^{n \times n}$ is the adjacency matrix of the edge-signed graph partially depicted above.  The ``eigenvalue upper bound'' $\EIG(\calI)$ arises from allowing the $x_i$'s to be arbitrary real numbers, subject to the constraint $\sum_i x_i^2 = n$.  The ``SDP upper bound'' $\SDP(\calI)$ (which is at least as tight: $\SDP(\calI) \leq \EIG(\calI)$) arises from allowing the $x_i$'s to be arbitrary unit vectors in $\R^n$, with the inner product $\la x_i, x_j \ra$ replacing $x_i x_j$ in the objective function.  Our goal is to identify some quantity $f(c)$ (it will be $\frac{1+\sqrt{2}}{2}$ in the $c = 2$ case) such that
\begin{equation}    \label{eqn:main-inequality}
    \EIG(\calI) \lesssim f(c) \lesssim \SDP(\calI)
\end{equation}
up to $1\pm o(1)$ factors, with high probability.  This establishes that all three quantities are equal (up to $1 \pm o(1)$, whp), since $\SDP(\calI) \leq \EIG(\calI)$ always.

In this section we mainly describe how to obtain the optimal inequality on the left in \pref{eqn:main-inequality}; i.e., how to give a tight bound on the eigenvalues of (the edge-signed graph induced by)~$\calI$.  Notice that if we were studying just random $\maxcut$ or $\twoxor$ CSPs, we would have to get tight bounds on the eigenvalues of a standard random $c$-regular graph.\footnote{More precisely, for random $\maxcut$ we have to lower-bound the smallest eigenvalue; for random $\twoxor$ --- which includes randomly negating edges --- we have to upper-bound the largest eigenvalue.  In the $\maxcut$ version with no negations, there is the usual annoyance that there is always a first ``trivial'' eigenvalue of~$c$, and one essentially wants to bound the second-largest (in magnitude) eigenvalue.  The effect of random negations is generally to eliminate the trivial eigenvalue, allowing one to  focus simply on the spectral radius of the adjacency matrix.  This technical convenience is one reason we will always work in a model that includes random negations.}  Excluding the top eigenvalue of~$c$ in the case of $\maxcut$, these eigenvalues are (whp) all at most $2\sqrt{c-1} + o_n(1)$ in magnitude.  This is thanks to Friedman's (difficult) proof of Alon's Conjecture~\cite{Fri08}, made moderately less difficult by Bordenave~\cite{bordenave2015new}.  The ``magic number'' $2\sqrt{c-1}$ is precisely the spectral radius of the \emph{infinite} $c$-regular tree --- i.e., the infinite graph that random $c$-regular graphs ``locally resemble''.

Returning to random $2$-constraint-regular instances of the $\sort_4$-CSP, the (edge-signed) infinite graph~$X$ that \emph{they} ``locally resemble'' is the following:

\myfig{0.25}{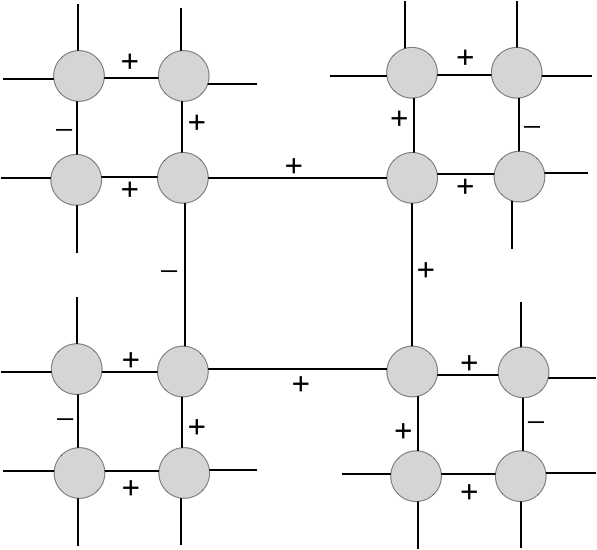}{$\sort_4$ infinite graph}{fig:sort-4-infinite}

Here $X \coloneqq \sort_4 \addprod \sort_4$ is the so-called \emph{additive product} of $2$ copies of the $\sort_4$ graph, a notion recently introduced in~\cite{MO18}.  By analogy with Alon's Conjecture, it's natural to guess that the spectral radius of a random $2$-constraint-regular $\sort_4$-CSP instance is whp $\rho(X) \pm o_n(1)$, where $\rho(X)$ denotes the spectral radius of~$X$ (which can be shown to be~$2\sqrt{2}$).  Indeed, our main effort is to prove the upper bound of $\rho(X) + o_n(1)$, thereby establishing the left inequality in \pref{eqn:main-inequality} with $f(c) = \rho(X)$.  (As for the right inequality, it can proven using the ``Gaussian Wave'' idea, allowing one to convert approximate eigenvectors of the infinite graph~$X$ to matching SDP solutions on random finite graphs~$\calI$.  We carry this out in \pref{sec:wave}.)

\subsubsection{Friedman/Bordenave Theorems for two-eigenvalue additive lifts} \label{sec:fried-bord}
As stated, our main task in the context of large random $2$-constraint-regular $\sort_4$-CSP instances is to show that their spectral radius is at most $\rho(X) + o_n(1)$ whp.  Incidentally, the lower bound of $\rho(X) - o_n(1)$ indeed holds; it follows from a generalization of the ``Alon--Boppana Bound'' due to Grigorchuk and \.{Z}uk~\cite{GZ99}.  As for the upper bound, the recent work~\cite{MO18} implies the analogous ``Ramanujan graph'' statement; namely, that there \emph{exist} arbitrarily large $2$-constraint-regular $\sort_4$-CSP instances with largest eigenvalue exactly upper-bounded by~$\rho(X)$.  However we need the analogue of Friedman/Bordenave's Theorem.  Unlike in~\cite{MO18} we are not able to prove it for arbitrary additive products; we are able to prove it for additive products of ``two-eigenvalue'' edge-signed graphs.  To explain why, we first have to review the proofs of the Alon Conjecture (that $c$-regular random graphs have their nontrivial eigenvalues bounded by $2\sqrt{c-1}+o_n(1)$).

Both Friedman's and Bordenave's proof of the Alon Conjecture rely on very sophisticated uses of the Trace Method.  Roughly speaking, this means counting closed walks of a fixed length~$k$ in random $c$-regular graphs, and (implicitly) comparing these counts to those in the $c$-regular infinite tree.  Actually, both works instead count only  \emph{nonbacktracking} walks.  The fact that one can relate nonbacktracking walk counts to general walk counts is thanks to an algebraic tool called the \emph{Ihara--Bass Formula} (more on which later); this idea was made more explicit in Bordenave's proof.  Incidentally, use of the nonbacktracking walk operator has played a major role in recent algorithmic breakthroughs on community detection and related results (e.g.,~\cite{KMMNSZZ13,MNS18,Mas14,BLM15}).

A reason for passing to nonbacktracking closed walks is that it greatly simplifies the counting.  Actually, in the case of the infinite $c$-regular tree, it \emph{over}simplifies the counting; infinite trees have no nonbacktracking closed walks at all!  However, the correct quantity to look at is  ``almost'' nonbacktracking walks of length~$k$, meaning ones that are nonbacktracking for the first $k/2$ steps, and for the last $k/2$ steps, but which may backtrack once right in the middle.  There are essentially $(c-1)^{k/2}$ of these in the $c$-regular infinite tree (one may take $k/2$ arbitrary steps out, but then one must directly walk back home), yielding a value of $((c-1)^{k/2})^{1/k} = \sqrt{c-1}$ for the spectral radius of the nonbacktracking operator of the $c$-regular infinite tree.  Bordenave uses (a very tricky version of) the Trace Method to analogously show that the spectral radius of the nonbacktracking operator of a random $c$-regular graph is $\sqrt{c-1} + o_n(1)$ whp.  Thanks to the Ihara--Bass Formula, this translates into a bound of $2\sqrt{c-1} + o_n(1)$ for the spectral radius of the usual adjacency operator.

Returning now to our scenario of random $2$-constraint-regular $\sort_4$-CSP instances (with their analogous infinite edge-signed graph~$X$), we encounter a severe difficulty.  Namely, passing to nonbacktracking walks no longer creates a drastic simplification in the counting, since there are nonbacktracking cycles within the constraint graphs themselves (in our example, $4$-cycles graphs).\footnote{In fact, since we have edge weights (signs), we need to look at the \emph{weight} (not number) of walks, but the point still stands.}  Thus nonbacktracking closed walks in large random instances can have complicated structures, with many internal nonbacktracking cycles.

A saving grace in the case of $\sort_4$-CSPs, and also ones based on $\forrelation_k$ or complete-graph constraints for example, is that the adjacency matrices of these graphs have only \emph{two distinct eigenvalues}.  (We will also use that their edge weights are~$\pm 1$.)  For example, after rearranging the variables in the $\sort_4$ predicate, its adjacency matrix is
\begin{equation}    \label{eqn:sort-adj}
    A = \begin{pmatrix} 0 & 0 & +1 & +1 \\ 0 & 0 & +1 & -1 \\ +1 & +1 & 0 & 0 \\ +1 & -1 & 0 & 0 \end{pmatrix},
\end{equation}
which has eigenvalues of $\pm \sqrt{2}$ (with multiplicity~$2$ each).  The two-eigenvalue property implies that~$A$ satisfies a quadratic equation, and hence any polynomial in~$A$ is equivalent to a polynomial of degree \emph{at most~$1$}.  The upshot is that we can relate general walks in $\sort_4$-CSPs (or more generally, CSPs with two-eigenvalue constraints) to what we call \emph{nomadic} walks: ones that take \emph{at most~$1$} consecutive step within a single constraint.  Let us make an informal definition (see \pref{sec:prelims-walks} for a formal definition):
\begin{definition}
    Given a finite CSP graph, the \emph{nomadic walk operator} $B$ is a matrix indexed by the directed edges in the graph.  Its $B[e,e']$ entry is equal to the edge-weight of~$e'$ provided:
        \begin{itemize}
            \item $(e,e')$ forms an oriented length-$2$ path; and,
            \item $e$ and $e'$ come from \emph{different} constraints.
        \end{itemize}
    Otherwise the $B[e,e']$ entry is~$0$.  This operator generalizes the nonbacktracking walk operator for $\maxcut$/$\twoxor$ graphs in which each undirected edge is considered to be a single ``constraint''.
\end{definition}

The utility of this nomadic walk operator is twofold for us.  First, for two-eigenvalue CSPs we can relate the eigenvalues of the usual adjacency operator to those of the nomadic walk operator through the following  generalization of the Ihara--Bass Formula:
\begin{theorem}[informal]             \label{thm:ihara-bass}
Let $A$ be the adjacency matrix and $B$ the nomadic walk operator of a finite $c$-constraint-regular CSP graph on $n$ vertices, where each predicate has exactly $2$ distinct eigenvalues: $\lambda_1$ and $\lambda_2$.  Define $L(t) := \Id - At + (\lambda_1+\lambda_2)t\Id + (c-1)(-\lambda_1\lambda_2) t^2$. Then we have
\[\displaystyle (1+\lambda_1t)^{n\frac{ c\lambda_2}{\lambda_2-\lambda_1}-1}(1+\lambda_2t)^{n\frac{c\lambda_1}{\lambda_1-\lambda_2}-1} \det L(t) = \det(\Id - Bt).\]
\end{theorem}
\noindent We prove \pref{thm:ihara-bass} in \pref{sec:ihara}.  In the remaining discussion below, we let $B$ be the nomadic walk operator of a random $c$-constraint-regular CSP graph on $n$ vertices, where the precise random model is given in \pref{def:random-model-precise}.  Further, we assume that the predicate of the CSP has two distinct eigenvalues: $\lambda_1$ and $\lambda_2$.

The second utility of nomadic walks  is that they provide the key simplification needed to make closed-walk counting in non-tree-like CSPs tractable.  Because of this, we are able to establish the following modification of Bordenave's proof of Friedman's Theorem in \pref{sec:bordenave}:

% XXXXXSTATE OUR BORDENAVEXXXXXXX
\begin{theorem} \label{thm:bordenave-informal}
    With high probability,
    \[
        \rho(B) \le \sqrt{(c-1)(-\lambda_1\lambda_2)} + o_n(1).
    \]
\end{theorem}

And we can use our version of Ihara--Bass, \pref{thm:ihara-bass}, to conclude bounds on the spectrum of the adjacency matrix $A$ from \pref{thm:bordenave-informal}, which is worked out in \pref{sec:bound}.
\begin{theorem} \label{thm:adjacency-spec-bounds-informal}
    With high probability,
    \[
        \spec(A) \subseteq \left[\lambda_1+\lambda_2-2\sqrt{(c-1)(-\lambda_1\lambda_2)}-o(1),\lambda_1+\lambda_2+2\sqrt{(c-1)(-\lambda_1\lambda_2)}+o(1)\right].
    \]
\end{theorem}

Yet another advantage of using nomadic walks instead of closed walks is that in \pref{thm:adjacency-spec-bounds-informal} we are able to bound the left and right spectral edge of $A$ by \emph{different} values, whereas counting closed walks would, at best, only give an upper bound on $|\lambda|_{\max}(A)$.

\pref{thm:adjacency-spec-bounds-informal} lets us conclude an upper bound on the SDP value, and we complement that with a lower bound via the construction of an SDP solution that nearly matches the upper bound.  In particular, we prove the following in \pref{sec:wave}.
\begin{theorem}
    For every $\eps > 0$, whp there exists a PSD matrix $M$ with an all-ones diagonal such that
    \[
        \langle A, M\rangle \ge \left(\lambda_1+\lambda_2+2\sqrt{(c-1)(-\lambda_1\lambda_2)}-\eps\right)n.
    \]
\end{theorem}

As detailed out in \pref{sec:wrapup}, this lets us conclude the main theorem of this paper:
\begin{theorem} \label{thm:punchline}
    For random $c$-constraint-regular instances of a CSP with $2$ distinct eigenvalues $\lambda_1$ and $\lambda_2$, the SDP value is in the range
    \[
        \frac{\lambda_1+\lambda_2+2\sqrt{(c-1)(-\lambda_1\lambda_2)}}{c(-\lambda_1\lambda_2)}\pm\eps
    \]
    with high probability, for any $\eps > 0$.
\end{theorem}

\pref{thm:forr} can be viewed as a special case of \pref{thm:punchline}.

\subsection{Relationship to the work of Bordenave--Collins} \label{sec:bc}
Xinyu Wu has brought to our attention the relevance to our work of a recent paper by Bordenave and Collins~\cite{bordenave2018eigenvalues}.  Briefly put, their paper establishes a Friedman/Bordenave theorem for large random graphs whose adjacency matrices are noncommutative polynomials in a fixed number of independent random matching matrices and permutation matrices (together with their transposes).  As a most basic example, it recovers the following form of Friedman's Theorem: whp, the sum of $d$ random perfect matchings has all nontrivial eigenvalues bounded in magnitude by  $\rho(\Z_2 \ast \cdots \text{($d$ times)} \cdots  \ast \Z_2) + o_n(1) = 2\sqrt{d-1} + o_n(1)$.  However, the Bordenave--Collins work gives much more than this.  For example, let $\bG$ be the $n$-vertex graph formed as 
\[
    \bP + \bP^\top + \boldM - \bP \boldM \bP^\top,
\]
where $\boldM$ is a random matching matrix and and $\bP$ is an independent random permutation matrix.  It is not hard to see that $\bG$ will essentially ``locally resemble'' a $2$-constraint-regular $\sort_4$-CSP instance. And, the Bordenave--Collins work implies that the eigenvalues of~$\bG$ are bounded (whp) by $\rho(\sort_4 \addprod \sort_4)$. Using the theory of free probability, it is possible to directly compute that $\rho(\sort_4 \addprod \sort_4) = 2\sqrt{2}$.  In this way, our \pref{thm:adjacency-spec-bounds-informal} in the case of $2$-constraint-regular $\sort_4$-CSPs is covered by Bordenave and Collins.  Indeed, it is not hard to generalize this example to the case of $c$-constraint-regular $\sort_4$-CSPs for any \emph{even} integer~$c$.

Indeed, the Bordenave--Collins work also treats some kinds of graphs that our work cannot; for example, Wu gave the example when $\bG$ is the $n$-vertex graph generated by the polynomial
\[
    \bP_1 + \bP_1^\top + \bP_2 +\bP_2^\top + \bP_3 + \bP_3^\top + \bP_4 + \bP_4^\top + \bP_1 \bP_2 \bP_3 \bP_4 + \bP_4^\top \bP_3^\top \bP_2^\top \bP_1^\top,
\]
where $\bP_1, \dots, \bP_4$ are independent uniformly random permutation matrices. This $\bG$  ``locally resembles'' the infinite free product graph $X = \Z_4 \ast \Z_4 \ast \Z_4 \ast \Z_4$, and the Bordenave--Collins work implies that whp, $\bG$'s nontrivial eigenvalues are bounded in magnitude by~$\rho(X) + o_n(1)$.  (We remark that computing the numeric value of this $\rho(X)$ is difficult, but possible; see, e.g.,~\cite[Ch.~9C]{woess2000random}).  Since the $4$-cycle graph $\Z_4$ has more than two distinct eigenvalues, it is not covered by our work.

This said, the Bordenave--Collins work does not subsume our \pref{thm:adjacency-spec-bounds-informal}, as there are plenty of graph families that our theorem handles but  Bordenave--Collins's does not (seem to).  For example, Wu has sketched to us a proof that one cannot obtain $c$-constraint-regular $\sort_4$ instances for \emph{odd}~$c$ through any straightforward use of \cite{bordenave2018eigenvalues}.  Additionally, even in the cases of interest to us where Bordenave--Collins applies, we can point to some  (minor) advantages of our methods.  For one, our model of random graph generation clearly corresponds to precisely-regular CSP instances, whereas in the Bordenave--Collins model there will be (in expectation) a constant number of local ``blemishes'' where one cannot interpret a piece of the graph as a constraint.  For another, our work directly yields the numerical values of the appropriate spectral radii~$\rho(X)$ (though in the cases where our results apply, these can be obtained through standard methods in free probability).

%% file: prelims.tex
\section{Preliminaries} \label{sec:prelims}

\subsection{$\twoxor$ optimization problems and their relaxations}
All of the CSPs studied in this work ($\maxcut$, $\naethree$, $\sort_4$, $\forrelation_k$, etc.)\ will effectively reduce to \emph{$\twoxor$ optimization problems} --- equivalently, the problem maximizing a homogeneous degree-$2$ polynomial with $\pm 1$ coefficients over the Boolean hypercube.
\begin{definition} (Optimization of $\twoxor$ instances)
    Let $G = (V,E)$ be an undirected graph (possibly with parallel edges), with edge-signing $\Wt : E \to \{\pm 1\}$.  We call the pair $\instgraph = (G, \Wt)$ an \emph{instance}. The associated \emph{$\twoxor$ optimization problem} is to determine the \emph{(true) optimum value}
    \[
        \OPT(\instgraph) = \max_{x : V \to \{\pm 1\}} \avg_{e = \{u,v\} \in E} \braces*{\Wt(e) x_u x_v} \in [-1,+1].
    \]
    The special case in which $\Wt \equiv -1$ is referred to as the $\maxcut$ problem on~$G$, as in this case $\frac12 + \frac12 \OPT(\instgraph) = \maxcut(G)$, the maximum fraction of edges that can be cut by a bipartition of~$V$.
\end{definition}
Determining $\OPT(\instgraph)$ is $\NP$-hard in the worst case, leading to the study of computationally tractable approximations/relaxations.  Two such approximations are the \emph{eigenvalue bound} and the \emph{SDP bound}, which we now recall.
\begin{definition} (Adjacency matrix/operator)
    The \emph{adjacency matrix} $A$ of a finite weighted graph $(G, \Wt)$ has rows and columns indexed by~$V$; the entry $A[u,v]$ equals the sum of $\Wt(e)$ over all edges with endpoints $\{u,v\}$.  In case $G$ is infinite we can more generally define the adjacency operator $A$ on $\ell_2(V)$ as follows:
    \[
        \text{for } F \in \ell_2(V), \quad AF(u) = \sum_{e = (u,v) \in E} \Wt(e) F(v).
    \]
\end{definition}
\begin{definition} (Eigenvalue bound)
    The \emph{eigenvalue bound} $\EIG(\instgraph)$ for $\twoxor$ instance $\instgraph$ with adjacency matrix~$A$ is $\frac{n}{2|E|} \lambda_{\textnormal{max}}(A)$, where  $\lambda_{\textnormal{max}}$ denotes the maximum eigenvalue.  We have $\OPT(\instgraph) \leq \EIG(\instgraph)$ always, as the eigenvalue bound captures the relaxation of $\twoxor$ optimization where we allow any $x : V \to \R$ satisfying $\|x\|^2 = n$.
\end{definition}

The \emph{SDP value} provides an even tighter upper bound on $\OPT(\instgraph)$, and is still efficiently computable.\footnote{More precisely, it can be computed to within $\pm \epsilon$ in $\poly(|\calI|, \log(1/\eps))$ time using the Ellipsoid Algorithm~\cite{GLS88,DP93}.} The SDP bound dates back to Lov{\'a}sz's Theta Function in the context of the $\indepset$ problem~\cite{Lov79}, and was proposed in the context of the $\maxcut$ problem by Delorme and Poljak~\cite{DP93}.
\begin{definition} (SDP bound)
    The \emph{SDP bound} $\SDP(\instgraph)$ for $\twoxor$ instance $\instgraph$ is
    \[
        \SDP(\instgraph) = \max_{\vec{x} : V \to S^{m-1}} \avg_{e = \{u,v\} \in E} \braces*{\Wt(e) \langle \vec{x}_u, \vec{x}_v\rangle} \in [-1,+1],
    \]
    where $S^{m-1}$ refers to the set of unit vectors in~$\R^m$ and the maximum is also over~$m$ (though $m = n$ is sufficient).   The following holds for all~$\instgraph$:
    \[
        \OPT(\instgraph) \leq \SDP(\instgraph) \leq \EIG(\instgraph).
    \]
    The left inequality is obvious.  One way to see the right inequality is to use the fact~\cite{DP93}, based on SDP duality, that $\SDP(\instgraph)$ is also equal to the minimum value of the eigenvalue bound applied to $A + Y$, where $A$ is the adjacency matrix and $Y$ ranges over all matrices of trace~$0$.
\end{definition}

Goemans and Williamson~\cite{GW95} famously showed that
\[
    \tfrac12 + \tfrac12\SDP(\calI) \leq 1.138 (\tfrac12 + \tfrac12\OPT(\calI))
\]
holds for every $\twoxor$ instance, and Feige--Schechtman~\cite{FS02} showed their bound can be tight in the worst case.\footnote{The case of $\maxcut$ on the $5$-cycle --- i.e., maximizing $-\frac15(x_1 x_2 + x_2 x_3 + x_3 x_4 + x_4 x_5 + x_5 x_1)$ on $\{\pm 1\}^5$ --- already has $\OPT = 3/5$ and $\SDP = (1+\sqrt{5})/4$, showing that $1.138$ cannot be improved below $1.131$.}  As for directly comparing $\SDP(\calI)$ and $\OPT(\calI)$, we have the following:
\begin{itemize}
    \item (\cite{CW04}) $\SDP(\calI) \leq O(\OPT(\calI) \cdot \log(1/\OPT(\calI)))$ always holds.
    \item When~$G$ is bipartite (a special case of particular interest, see \pref{sec:quantum}), it holds that $\SDP(\calI) \leq K \cdot \OPT(\calI)$ for constant~$K$. This is known as \emph{Grothendieck's inequality}~\cite{Gro53}, and the constant is known~\cite{BMMN13} to satisfy $K < \pi/(2 \ln(1+\sqrt{2})) \approx 1.78$.
\end{itemize}

\subsection{Quantum games, and some quantum-relevant constraints}  \label{sec:quantum}
In the case when the underlying graph $G$ is bipartite, $\SDP(\calI)$ has another important interpretation: it is the true \emph{quantum} value of the $2$-player $1$-round  ``nonlocal game'' associated to~$\calI$.  We give definitions below, but let us mention that the $\sort_4$ (equivalently, $\chsh$) and $\forrelation_k$ constraints from \pref{thm:sort} and \pref{thm:forr} are both: (a)~bipartite; (b)~directly inspired by quantum theory.  Thus those two theorems can be interpreted as determining the true quantum value of random $c$-constraint-regular nonlocal games based on $\chsh$ and $\forrelation_k$.

Let us now recall the relevant quantum facts.
\begin{definition}[Nonlocal $\twoxor$ games]
    Given a $\twoxor$ instance $\calI = (G,\Wt)$ with $G = (U,V,E)$ bipartite, the associated \emph{nonlocal ($\twoxor$) game} is the following.  There are spatially separated players Alice and Bob.  A referee chooses $e = (u,v) \in E$ uniformly at random, tells~$u$ to Alice, and tells~$v$ to Bob.  Without communicating, Alice and Bob are required to respond with signs~$x_u, y_v \in \{\pm 1\}$.  The \emph{value} to the players is the expected value of $\Wt(e) x_u y_v$.  It is easy to see that if Alice and Bob are deterministic, or are allowed classical shared randomness, then the optimum value they can achieve is precisely~$\OPT(\calI)$.
\end{definition}
\begin{theorem}                                     \label{thm:quantumvalue}
    (\cite{CHTW04,Tsi80}.) In a nonlocal $\twoxor$ game, if Alice and Bob are allowed to share unlimited quantumly entangled particles, then the optimal value they can achieve is precisely~$\SDP(\calI)$.
\end{theorem}
The fact that there exist bipartite edge-signed~$\calI$ for which $\SDP(\calI) > \OPT(\calI)$ is foundational for the experimental verification of quantum mechanics, as the following example attests:
\begin{example}
    Consider the $\twoxor$ instance depicted in \pref{fig:chsh}, called $\chsh$ after Clauser, Horne, Shimony, and Holt~\cite{CHSH69}.    It has
    \[
        \OPT(\CHSH) = 1/2 < 1/\sqrt{2} = \SDP(\CHSH).
    \]
     \myfig{.25}{chsh.pdf}{The CHSH game/CSP}{fig:chsh}
    The upper bound $4 \cdot \OPT(\CHSH) \leq 2$ is often called \emph{Bell's inequality}~\cite{Bel64}, and the higher lower bound $1/\sqrt{2} \leq \SDP(\CHSH)$ is from~\cite{CHSH69} (with $\SDP(\CHSH) \leq 1/\sqrt{2}$ due to Tsirelson~\cite{Tsi80}).  Aspect and others~\cite{ADR82} famously experimentally realized this gap between what can be achieved with classical vs.\ quantum resources.
\end{example}
In fact, the $\CHSH$ instance is nothing more than the $\sort_4$ predicate in disguise! More precisely (cf.~\pref{eqn:sort-adj}),
\[
    \CHSH(x_1, x_2, x_3, x_4) = \tfrac14(x_1 x_3 + x_2 x_3 + x_1 x_4 - x_2 x_4) =  \sort_4(x_2, x_3, x_1, x_4) - \tfrac12.
\]
Thanks to its degree-$2$ Fourier expansion, CSPs based on the $\sort_4$/$\CHSH$ constraint have been studied in a variety of contexts, including concrete complexity~\cite{Amb06,APV16,OSTWZ14} and fixed parameter algorithms~\cite{Wil07}.

Though $\sort_4$ is a ``predicate'', in the sense that it takes $0$/$1$ (unsat/sat) values, there's nothing necessary about basing a large CSP on predicates.  An interesting family of constraints that can be modeled by $\twoxor$ optimization, originally arising in quantum complexity theory~\cite{AA15}, is the family of  ``Forrelation'' functions.  For any $k \in \N$, the $\Forrelation_{k}$ function is defined by
\[
    \Forrelation_k : \{\pm 1\}^{2^k} \times \{\pm 1\}^{2^k} \to [-1,+1], \qquad
    \Forrelation_k(x_1, \dots, x_{2^k}, y_1, \dots, y_{2^k}) = 2^{-2k} x^\top H_k y,
\]
where $H_k = \begin{pmatrix} +1 & +1 \\ +1 & -1 \end{pmatrix}^{\otimes k}$ is the $k$th Walsh--Hadamard matrix.  Note that $\Forrelation_0$ corresponds to the single-(positive-)edge $\twoxor$ CSP, and $\Forrelation_1$ is $\CHSH$. %  [[XXX I would like to now say what $\OPT$ and $\SDP$ are for $\Forrelation_k$; theory of ``bent'' functions comes in.  I'm 80\% pretty sure that $\OPT(\Forrelation_k) = 2^{-\lceil k/2 \rceil}$ and $\SDP(\Forrelation_k) = 2^{-k/2}$. XXX]]

\subsection{$\twoxor$ graphs with only $2$ distinct eigenvalues}
As mentioned, the class of constraints that we treat in this work are those that can be modeled as $\twoxor$ instances with \emph{$2$ distinct eigenvalues}.  The $\Forrelation_k$ constraint is a prime example; when viewed as an edge-signed graph (i.e., ignoring the $2^{-2k}$ scaling factors), its eigenvalues are all $\pm 2^{k/2}$.  Another example is the complete graph constraint on~$r$ variables, which has eigenvalues of~$r-1$ and $-1$ (the latter with multiplicity~$r-1$).  The $r = 3$ complete-graph case, after a trivial affine shift, also corresponds to a Boolean predicate that is well known in the context of CSPs: the $\naethree$ predicate, as studied in~\cite{deshpande2018threshold}.  This is because
\[
    \naethree(x_1, x_2, x_3) = \frac34 - \frac34(x_1 x_2 + x_2 x_3 + x_3 x_1).
\]

Let us make some definitions we will use throughout the paper.
\begin{definition}[$2$-eigenvalue graphs]   \label{def:2}
    We call an undirected, edge-weighted simple graph $\calI$ a \emph{$2$-eigenvalue graph} if there are two real numbers $\lambda_1$ and $\lambda_2$ such that each eigenvalue of $\calI$'s (signed) adjacency matrix $A$ is equal to either $\lambda_1$ or $\lambda_2$.
\end{definition}
See, e.g., \cite{Ram15} for a paper studying such graphs.  In this section, let us use the notation from \pref{def:2} and prove some properties that will be used throughout the paper.

First, since $A$ is symmetric, its eigenvectors are spanning and therefore every vector can be written as the sum of a vector in $\ker(A - \lambda_1 \Id)$ and one in $\ker(A - \lambda_2 \Id)$.  Thus:
\begin{proposition} \label{prop:minimal-poly}
 $(A-\lambda_1\Id)(A-\lambda_2\Id) = 0$, where $\Id$ denotes the identity matrix.
\end{proposition}
%\begin{proof}
%    We can write any vector $v$ as $v_1+v_2$, where $v_1$ is in the eigenspace corresponding to $\lambda_1$ and $v_2$ is in the eigenspace corresponding to $\lambda_2$.  The result follows from considering the following.
%    \begin{align*}
%        (A_G-\lambda_1\Id)(A_G-\lambda_2\Id)v &= (A_G-\lambda_1\Id)(\lambda_1-\lambda_2)v_1\\
%        &= 0
%    \end{align*}
%\end{proof}
%

%Now, observe that \pref{prop:minimal-poly}
This proposition
implies that $A^2 = (\lambda_1+\lambda_2)A - \lambda_1\lambda_2\Id$.  Thus we can deduce the following two facts:
\begin{fact}    \label{fact:diag-square}
    For any $v\in V(G)$, $\displaystyle \sum_{u\in V(G)} A[u,v]^2  = A^2[v,v] = -\lambda_1\lambda_2$.
\end{fact}

\begin{fact}    \label{fact:off-diag-square}
    For any pair of distinct vertices $u,v\in V(G)$, \[\sum_{w\in V(G)} A[u,w] A[w,v] = A^2[u,v] = (\lambda_1+\lambda_2)A[u,v].\]
\end{fact}

\subsection{Random constraint graphs, instance graphs, and additive products} \label{sec:prelims-walks}
\begin{definition}[Constraint graphs]
    An \emph{$r$-ary, $c$-atom constraint graph} is any $n$-fold lift~$\congraph$ of the complete bipartite graph $K_{r,c}$.  Each vertex on the $c$-regular side is called a \emph{variable vertex}, and is typically depicted by a circle. The variable vertices are partitioned into~$r$ \emph{variable groups} each of size~$n$, called the \emph{$1$st variable group}, the \emph{$2$nd variable group}, etc.  Each vertex on the $r$-regular side is called a \emph{constraint} (or \emph{atom}) \emph{vertex}, and is typically depicted by a square.  Again, the constraint vertices are partitioned into $c$ \emph{constraint}  (or \emph{atom}) groups of size~$n$, called the \emph{$1$st constraint/atom group}, \emph{$2$nd constraint/atom group}, etc.  When $n = 1$, we call $\congraph$ a \emph{base constraint graph}.  We also allow ``$n = \infty$'': this means we take the infinite $(r,c)$-biregular tree and partition its variable vertices into $r$~groups and its constraint variables into $c$~groups in such a way that every variable vertex in the $i$th group has exactly one neighbor from each of the $c$ constraint groups, and similarly every constraint vertex in the $j$th group has exactly one neighbor from each of the $r$ variable groups. An example of a constraint graph is shown in \pref{fig:instcons}. \footnote{This can be done in an arbitrary ``greedy'' way, fixing any, say, constraint vertex to be in ``group~$1$'', fixing its variables neighbors to be in groups $1 \dots r$ in an arbitrary way, fixing \emph{their} constraint neighbors to be in groups $2 \dots c$ in an arbitrary way, etc.}
\end{definition}

\myfig{.25}{example_krc.pdf}{The complete $K_{4,3}$ graph}{fig:krc}

\begin{definition}[Instance graphs] \label{def:inst-graph}
    Let $\atomlist = (A_1, \dots, A_c)$ be a sequence of \emph{atoms}, meaning edge-weighted undirected graphs on a common vertex set~$[r]$.  (In this paper, the edge-weights will usually be~$\pm 1$.)  We also think of each atom as a collection of ``$\twoxor$-constraints'' on variable set~$r$.  Now given an $r$-ary, $c$-atom constraint graph~$\congraph$, we can combine it with the atom specification $\atomlist$ to form the \emph{instance graph} $\instgraph \coloneqq \atomlist(\congraph)$.  This edge-weighted undirected graph $\instgraph$ has as its vertex set all the variable vertices of~$\congraph$.  The edges of $\instgraph$ are formed as follows: We iterate through each $j \in [c]$ and each constraint vertex~$f$ in the $j$th constraint group of~$\congraph$. Given $f$, with variables neighbors $v_1, \dots, v_r$ in $\congraph$, we place a copy of atom $A_j$ onto these vertices in~$\instgraph$.  ($\instgraph$ may end up with parallel edges.)  We refer to the graph obtained by placing a copy of $A_j$ on vertices $v_1,\dots,v_r$ as $A_f$, and for any edge $e$ in $\instgraph$ that came from placing $A_j$, we define $\Atom(e) \coloneqq A_f$.  We use $v\sim A_f$ to denote that $v$ is one of $v_1,\dots,v_r$.  For $u,v\in\{v_1,\dots,v_r\}$, $A_f(u,v)$ denotes the edge in $A_f$ between $u$ and $v$.  And finally, denote the set $\{A_f:f~\text{constraint vertex in}~\congraph\}$ with $\Atoms(\instgraph)$. An example of an instance graph and corresponding constraint graph is shown in \pref{fig:instcons}.
\end{definition}

\begin{remark}
    Forming $\instgraph$ from $\congraph$ is somewhat similar to squaring $\congraph$ (in the graph-theoretic sense) and then restricting to the variable vertices.  With this in mind, here is an alternate way to describe the edges of $\instgraph$:  For each pair of distinct vertices $v, v'$ in $\instgraph$ (in variable groups $i$ and $i'$, respectively) we consider all length-$2$ paths joining~$v$ and~$v'$ in $\congraph$.  For each such path passing through a constraint vertex in constraint group~$j$, we add the edge $(v, v')$ into $\instgraph$ with edge-weight $A_j[i,i']$ (which may be~$0$).
\end{remark}

\begin{remark}  \label{rem:atom-transitive}
    We treat atoms as edge-weighted, undirected, complete graphs.  Thus, for a constraint vertex $f$ in constraint-graph $\congraph$, if there is an edge between vertices $u$ and $v$, and an edge between vertices $v$ and $w$ in the atom $A_f$, then there is an edge between $u$ and $w$ in $A_f$.  This view is significant in light of the proof of \pref{thm:ihara-bass-restatement}.
\end{remark}

\begin{figure}[!ht]
  \centering
  \begin{subfigure}[b]{0.45\textwidth}
    \centering
    \includegraphics[width=0.9\textwidth]{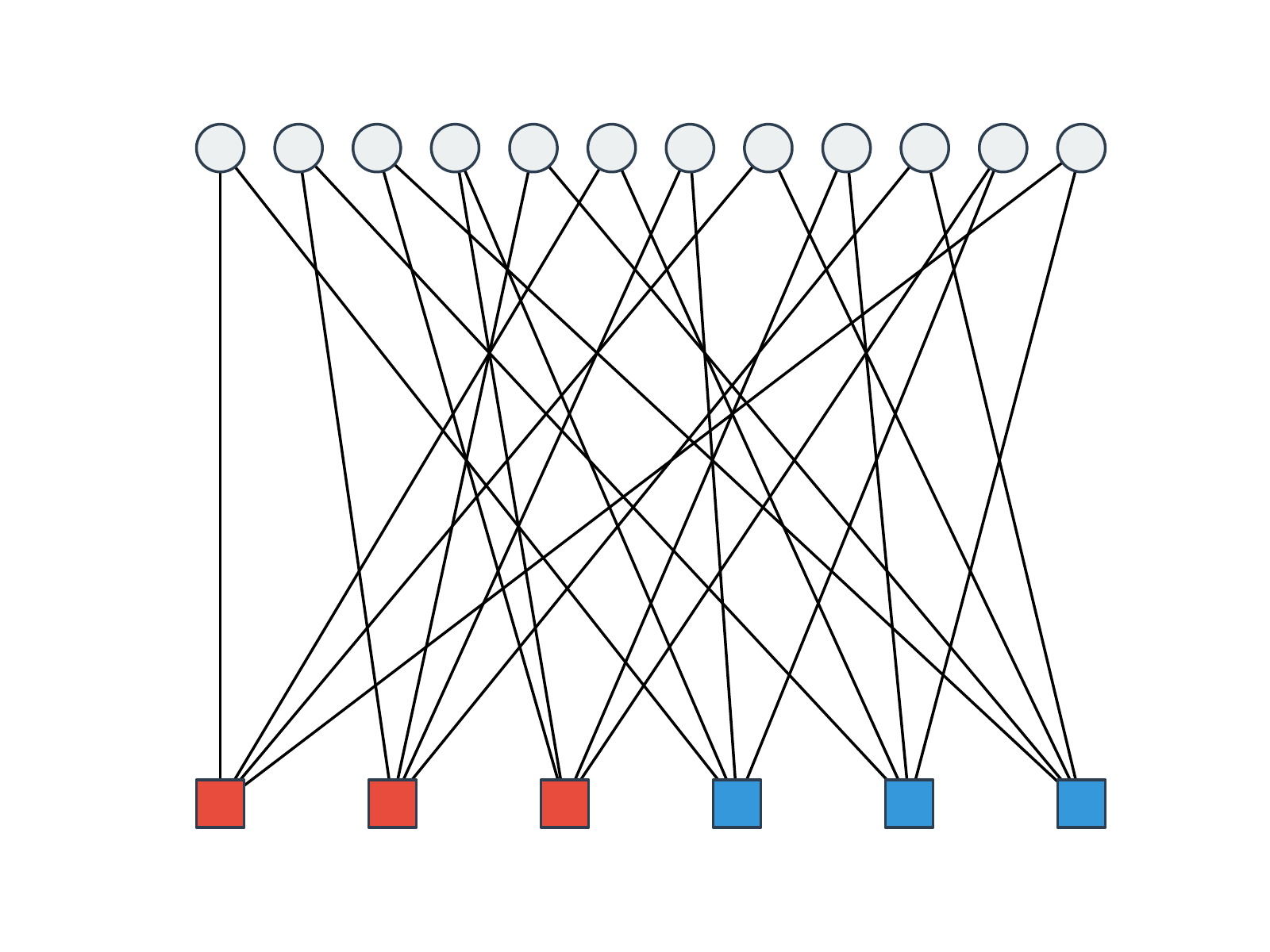}
  \end{subfigure}
  ~
  \begin{subfigure}[b]{0.45\textwidth}
    \centering
    \includegraphics[width=0.9\textwidth]{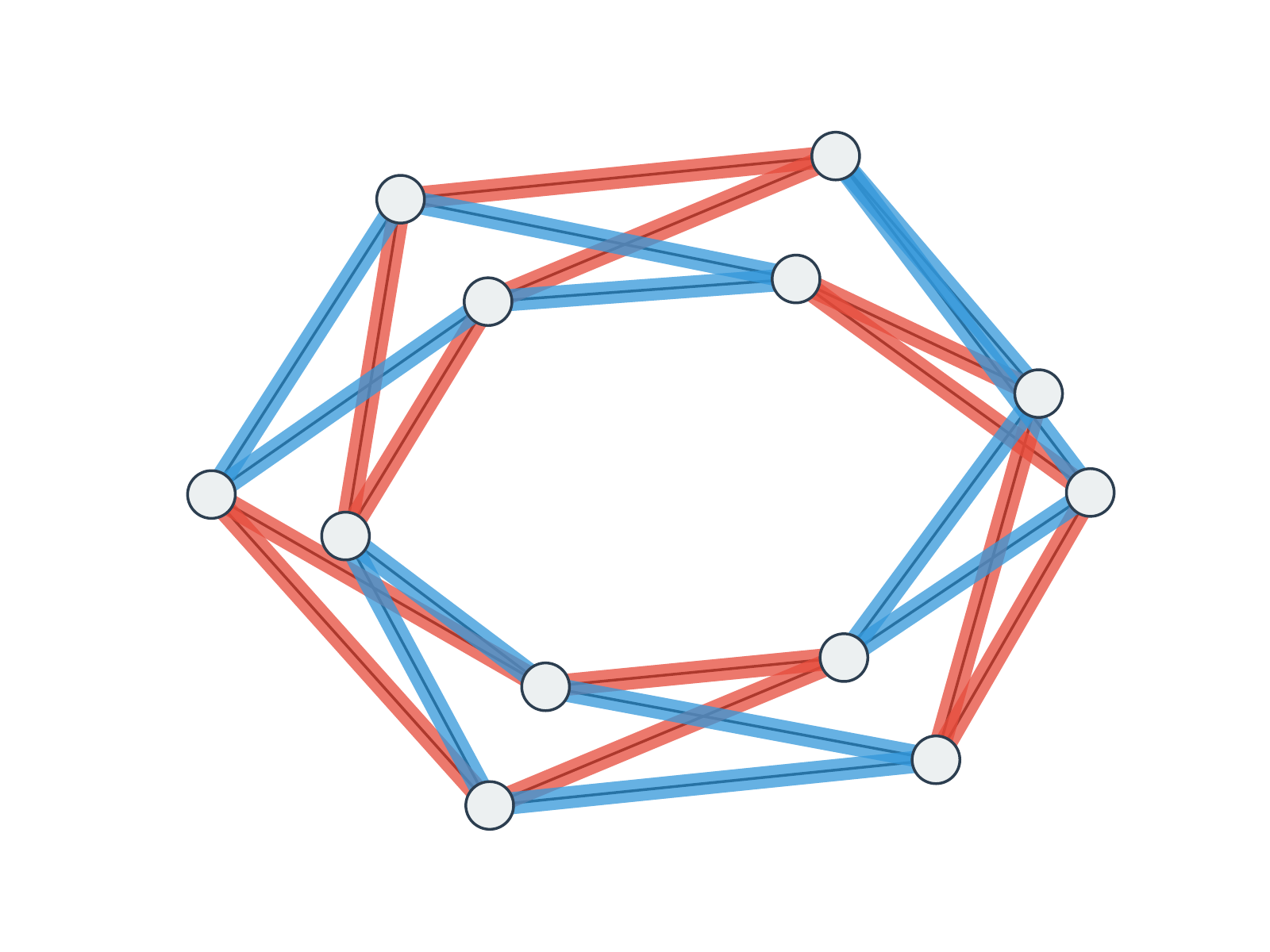}
  \end{subfigure}
  
  \caption{The figure on the left shows an example of a $4$-ary, $2$-atom $3$-fold lift constraint graph, with the left bipartition color coded by constraint/atom groups. The figure on the right is the corresponding instance graph on $(C_4, C_4)$, two four-cycle graphs, where each atom is color coded to match the figure on the left.}
  \label{fig:instcons}
\end{figure}

The following notions of additive lifts and additive products were introduced in~\cite{MO18}:
\begin{definition}[Random additive lifts]
    In the context of $r$-ary, $c$-atom constraint graphs, a \emph{random $n$-lifted constraint graph} simply means a usual random $n$-lift $\congraph$ (see, e.g.,~\cite{BL06}) of the base constraint graph.  Given atoms $\atomlist = (A_1, \dots, A_c)$, the resulting instance graph $\calI = \atomlist(\congraph)$ is called a \emph{random additive lift} of $\atomlist$.
\end{definition}
\begin{definition}[Additive products]
    If instead $\congraph$ is the ``$\infty$-lift'' of $K_{r,c}$, the resulting infinite instance graph $\calI = \atomlist(\congraph)$ is called the \emph{additive product} of $A_1, \dots, A_c$, denoted $A_1 \addprod A_2 \addprod \cdots \addprod A_c$.
\end{definition}

We will also extend \pref{def:inst-graph} to allow random additive lifts with \emph{negations}.  Eventually we will define a general notion of ``$1$-wise uniform negations'', but let us begin with two special cases.  In the ``constraint negation'' model, we assign to each constraint vertex~$f$ in $\congraph$ (from group~$j$) an independent uniformly random sign~$\xi^f$.  Then, when the instance graph~$\instgraph$ is formed from $\congraph$, each edge engendered by the constraint~$f$ has its weight multiplied by~$\xi^f$.  (Thus the edges in this copy of the atom~$A_j$ are either all left alone or they are simultaneously negated, with equal probability.)  In the ``variable negation'' model, for each group-$j$ constraint vertex~$f$, adjacent to variable vertices~$v_1, \dots, v_r$, we assign independent and uniformly random signs $(\xi^f_{i})_{i \in [r]}$ to the variables.  Then when the copy of~$A_j$ is added into~$\instgraph$, the $\{i,i'\}$-edge has its weight multiplied by $\xi^f_{i} \xi^f_{i'}$.  This corresponds to the constraint being applied to random \emph{literals}, rather than variables.

Notice that in both of these negation models, every time a copy of atom $A_j$ is placed into~$\instgraph$, its edges are multiplied by a collection of random signs $(\xi^f_{ij})_{i,j \in [r]}$ which are ``$1$-wise uniform''.  This is the only property we will require of a negation model.
\begin{definition}[Random additive lifts with negations]    \label{def:random-model-precise}
     A random additive lift \emph{with $1$-wise uniform negations} is a variant of \pref{def:inst-graph} where, for each constraint vertex~$f$ there are associated random signs $\xi^{(f)}_{i} \in \{\pm 1\}$, where $i \in [r]$.  For each fixed~$f$, the random variables $\xi^{(f)}_{i}$ are required to be $\pm 1$ with probability $1/2$ each, but they may be arbitrarily correlated; across different~$f$'s, the collections $(\xi^{(f)}_{i})_{i\in[r]}$ must be independent.  When the instance graph~$\instgraph$ is formed as $\atomlist(\congraph)$, and a copy of $A_j$ placed into $\instgraph$ thanks to constraint vertex~$f$, each new edge $\{i,i'\}$ has its weight $A_j[i,i']$ multiplied by $\xi^{(f)}_{ii'}:=\xi^{(f)}_{i}\xi^{(f)}_{i'}$.
\end{definition}

\begin{remark}
    For a given constraint-vertex $f$ of an instance graph $\instgraph$ obtained via a random additive lift with negations, the matrix $\mathrm{Adj}(A_f)$ has the same spectrum as $\mathrm{Adj}(\overline{A_f})$ where $\overline{A_f}$ denotes the subgraph prior to applying random negations, since there is a sign diagonal matrix $D$ such that $\mathrm{Adj}(\overline{A_f})=D\cdot\mathrm{Adj}(A_f)\cdot D^{\dagger}$.
\end{remark}

\subsection{Nomadic walks operators}

\begin{definition}[Nomadic walks]
    Let $\congraph$ be a constraint graph, $\atomlist = (A_1, \dots, A_c)$ a sequence of atoms, and $\instgraph = \atomlist(\congraph)$ the associated instance graph.  For initial simplicity, assume the atoms are unweighted (i.e., all edge weights are~$+1$).  A \emph{nomadic walk} in $\instgraph$ is a walk where consecutive steps are prohibited from ``being in the same atom''.  Note that if $r = 2$ and the atoms are single edges, a nomadic walk in $\instgraph$ is equivalent to a nonbacktracking walk.

    To make the definition completely precise requires ``remembering'' the constraint graph structure~$\congraph$.  Each step along an edge of $\instgraph$ corresponds to taking two consecutive steps in~$\congraph$ (starting and ending at a variable vertex).  The walk in $\instgraph$ is said to be nomadic precisely when the associated walk in~$\congraph$ is nonbacktracking.

    Finally, in the general case when the atoms~$A_j$ have weights, each \emph{walk} in~$\instgraph$ gets a weight equal to the product of the edge-weights used along the walk.
\end{definition}

\begin{figure}[!ht]
  \centering
  \begin{subfigure}[c]{0.45\textwidth}
    \centering
    \includegraphics[width=0.9\textwidth]{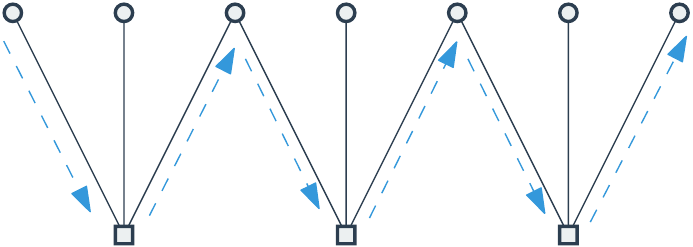}
  \end{subfigure}\hfill
  ~
  \begin{subfigure}{0.45\textwidth}
    \centering
    \includegraphics[width=0.9\textwidth]{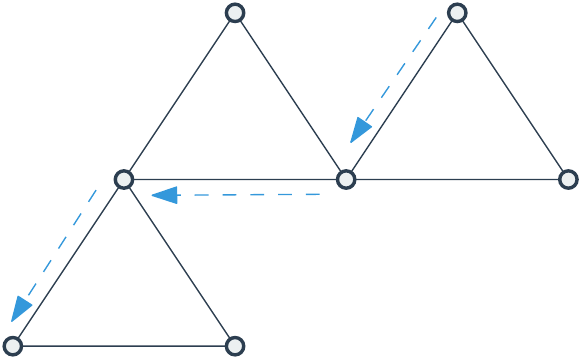}
  \end{subfigure}
  
  \caption{The figure on the left shows a nonbacktracking walk on a subset of a $3$-ary constraint graph and the one on the right the same nomadic walk on the corresponding instance graph.}
  \label{fig:nomadic}
\end{figure}

\begin{definition}[Nomadic walk operator]
    In the setting of the previous definition, the \emph{nomadic walk operator}~$B$ for $\instgraph$ is defined as follows.  Each edge $e = \{u,v\}$ in~$\instgraph$ is regarded as two opposing directed edges $\vec{e} = (u,v)$ and $\vec{e}^{-1} = (v,u)$, each having the same edge-weight as~$e$; i.e., $\weight(\vec{e}) = \weight(\vec{e}^{-1}) = \weight(e)$.  Let $\vec{E}$ denote the collection of all directed edges.  Now $B$ is defined to be the following linear operator on $\ell_2(\vec{E})$:
    \[
        \text{for $F \in \ell_2(\vec{E})$,} \quad B F(\vec{e}) = \sum_{\vec{e}'} \weight(\vec{e'}) F(\vec{e'}),
    \]
    where the sum is over all directed edges $\vec{e}'$ such that the pair $(\vec{e}, \vec{e}')$ forms a nomadic walk of length-$2$. In the finite-graph case we also think of~$B$ as a matrix; the entry $B[\vec{e},\vec{e}'] = \weight(\vec{e}')$ whenever $(\vec{e}, \vec{e}')$ is a length-$2$ nomadic walk.  Again, in the case where $r = 2$ and all atoms are single edges, the nomadic walk operator~$B$ coincides with the nonbacktracking walk operator.  (See, e.g.,~\cite{AFH15} for more on nonbacktracking walks operators.)
\end{definition}

\subsection{Operator Theory}
The results in this section can be found in a standard textbook on functional analysis or operator theory (see, for e.g. \cite{kubrusly2012spectral}).

Let $V$ be an some countable set and let $T:\ell_2(V)\rightarrow\ell_2(V)$ be a bounded, self-adjoint linear operator.
\begin{definition}
    We refer to the \emph{spectrum} of $T$, $\spec(T)$, as the set of all complex $\lambda$ such that $\lambda \Id-T$ is not invertible.  $\spec(T)$ is a nonempty, compact set.
\end{definition}

\begin{definition}
    We call $\lambda$ an \emph{approximate eigenvalue} of $T$ if for every $\eps>0$, there is unit $x$ in $\calX$ such that $\|Tx-\lambda x\|\le\eps$. We call such an $x$ an \emph{$\eps$-approximate eigenvector} or \emph{$\eps$-approximate eigenfunction}.
\end{definition}

\begin{theorem}
    If $T$ is a self-adjoint operator, then every $\lambda\in\spec(T)$ is an approximate eigenvalue.
\end{theorem}

\begin{theorem}\label{thm:isolation}[Consequence of Proposition 4.L of \cite{kubrusly2012spectral}]
    If $\lambda$ is an isolated point in $\spec(T)$, then it is an eigenvalue of $T$,  i.e., it is a $0$-approximate eigenvalue.
\end{theorem}

\begin{corollary}\label{cor:extremal_approx_eigens}
    $\lambda_{\min}:=\min\{\spec(T)\}$ and $\lambda_{\max}:=\max\{\spec(T)\}$ are both approximate eigenvalues of~$T$.
\end{corollary}

\begin{fact}    \label{fact:rayleigh}
    Additionally,
    \begin{align*}
        \lambda_{\min}(T) = \inf_{\|x\| = 1} \langle x, Tx\rangle,\\
        \lambda_{\max}(T) = \sup_{\|x\| = 1} \langle x, Tx\rangle.
    \end{align*}
\end{fact}

\begin{definition}
    The \emph{spectral radius} $\specrad(T)$ is defined as $\max_{\sigma\in\spec(T)} |\sigma|$.
\end{definition}

\begin{definition}
    The \emph{operator norm} of $T$, denoted $\|T\|_{\mathrm{op}}$, is defined as
    \[
        \sup_{\|x\| = 1, \|y\| = 1} \langle y, Tx \rangle = \sup_{\|x\| = 1}\|Tx\|.
    \]
\end{definition}

\begin{fact}    \label{fact:specrad-op-norm}
    \(\specrad(T) = \lim\limits_{k\to\infty}\|T^k\|_{\mathrm{op}}^{1/k}\).
\end{fact}

%% file: ihara.tex
\section{An Ihara--Bass formula for additive lifts of 2-eigenvalue atoms}  \label{sec:ihara}

Let $\calA$ be a sequence of atoms such that every atom has the same pair of exactly two distinct eigenvalues, $\lambda_1$ and $\lambda_2$, and let $\calH$ be a constraint graph on variable set $V$.  Let $\calI = \calA(\calH)$ be the corresponding instance graph.  In this section, we use $A$ and $B$ to refer to the adjacency matrix and nomadic walk matrix respectively of $\calI$.  The vertex set of $\instgraph$ is~$V$.  This section is devoted to proving our generalization of the Ihara--Bass formula, stated below.
\begin{theorem}                                     \label{thm:ihara-bass-restatement}
Let $L(t) := \Id - At + (\lambda_1+\lambda_2)t\Id + (c-1)(-\lambda_1\lambda_2) t^2$. Then we have
\[\displaystyle (1+\lambda_1t)^{|V|\frac{ c\lambda_2}{\lambda_2-\lambda_1}-1}(1+\lambda_2t)^{|V|\frac{c\lambda_1}{\lambda_1-\lambda_2}-1} \det L(t) = \det(\Id - Bt).\]
\end{theorem}

Our proof is a modification of one of the proofs of the Ihara--Bass formula from \cite{North97}.

\paragraph{Nomadic Polynomials.}  Our first step is to define the following sequence of polynomials.
\begin{align*}
    p_0(x) &= 1\\
    p_1(x) &= x\\
    p_2(x) &= x^2 - (\lambda_1+\lambda_2)x - c(-\lambda_1\lambda_2)\\
    p_k(x) &= xp_{k-1}(x) - (\lambda_1+\lambda_2)p_{k-1}(x) - (c-1)(-\lambda_1\lambda_2)p_{k-2}(x) &\text{for $k\ge 3$}
\end{align*}
and introduce the key player in the proof: the matrix of generating functions $F(t)$ defined by
\[
    F(t)_{u,v}   = \sum_{k \geq 0} p_k(A) t^k.
\]

\noindent We use $\Wt(e)$ to denote the weight on edge $e$, and define the weight of a walk $W = e_1e_2\dots e_\ell$ as
\[
    \Wt(W) := \prod_{i=1}^{\ell} \Wt(e_i).
\]

We first establish combinatorial meaning for the polynomials $p_k(A)$.
\begin{claim}
$p_k(A)_{uv}$ is equal to the total weight of nomadic walks of length $k$ from $u$ to $v$. 
\end{claim}
\begin{proof}
When $k = 0$ and $1$, the claim is clear.  We proceed by induction.

Supposing the claim is indeed true for $p_{s}(A)$ when $s\le k-1$, then $Ap_{k-1}(A)_{uv}$ is the total weight of length-$k$ walks from $u$ to $v$ whose first $k-1$ steps are nomadic and whose last step is arbitrary.  Call the collection of these walks $\calW_{uv}$.  For $W\in\calW_{uv}$, let $W_i$ denote the edge walked on by the $i$-th step of $W$ and let $W_{(i)}$ denote the length-$i$ walk obtained by taking the length-$i$ prefix of $W$.  We use lowercase $w_i$ to denote the vertex visited by the $i$th step of the walk.  Each $W\in\calW_{uv}$ falls into one of the following three categories.
\begin{enumerate}
    \item \label{type:nomadic} $W$ is a nomadic walk.  Call the collection of these walks $\calW_{uv}^{(1)}$.
    \item \label{type:backtrack} $W_{k} = W_{k-1}^{-1}$.  Call the collection of these walks $\calW_{uv}^{(2)}$.
    \item \label{type:hangout} $W_{k-1}$ and $W_k$ are in the same atom but $W_{k} \ne W_{k-1}^{-1}$.  Call the collection of these walks $\calW_{uv}^{(3)}$.
\end{enumerate}

\noindent Suppose $k\ge 3$.
\begin{align*}
    \sum_{W\in\calW_{uv}^{(2)}} \Wt(W) &= \sum_{W\in\calW_{uv}^{(2)}} \Wt(W_{k-1})\Wt(W_{k-1}^{-1}) \Wt(W_{(k-2)})\\
    &= \sum_{W\in\calW_{uv}^{(2)}} \Wt(W_{k-1})^2 \Wt(W_{(k-2)})\\
    &= \sum_{\substack{W'~\text{$(k-2)$-length nomadic walk}\\ \text{from $u$ to $v$}}} \Wt(W')\sum_{e\notin\Atom(W'_{k-2})} \Wt(e)^2
    \intertext{We apply \pref{fact:diag-square} and get}
    &= \sum_{\substack{W'\text{ $(k-2)$-length nomadic walk}\\ \text{from $u$ to $v$}}}\Wt(W')(c-1)(-\lambda_1\lambda_2)\\
    &= (c-1)(-\lambda_1\lambda_2)p_{k-2}(A)_{uv}.
\end{align*}
An identical argument shows that when $k = 2$,
\[
    \sum_{W\in\calW_{uv}^{(2)}} \Wt(W) = c(-\lambda_1\lambda_2)
\]

We do a similar calculation for $\calW_{uv}^{(3)}$ for $k\ge 2$.  Observe that $W_{k-1}$ and $W_k$ have to be in the same atom, which we denote $\Atom(W_{k-1})$.  Thus, there is an edge $e^*$ between $w_{k-2}$ and $v$ in $\Atom(W_{k-1})$ too (see \pref{rem:atom-transitive}).
\begin{align*}
    \sum_{W\in\calW_{uv}^{(3)}} \Wt(W) &= \sum_{W\in\calW_{uv}^{(3)}} \Wt(W_{k-1})\Wt(W_k)\Wt(W_{(k-2)})\\
    &=\sum_{\substack{W'~\text{length-$(k-2)$ nomadic walk}\\ W'_0 = u,\\\text{$e^*$ s.t. $(e^*)_1=w_{k-2},(e^*)_2=v$}\\ \Atom(W'_{k-2})\ne\Atom(e^*)}}~\sum_{\substack{e^{(1)},e^{(2)}:\\ \Atom(e^{(1)})=\Atom(e^{(2)}) = \Atom(e^*)\\ (e^{(1)})_1=w_{k-2}, (e^{(1)})_2=(e^{(2)})_1, (e^{(2)})_2=v }} \Wt(e^{(1)})\Wt(e^{(2)})\Wt(W')
    \intertext{By applying \pref{fact:off-diag-square}, we get}
    &= \sum_{\substack{W'~\text{length-$(k-2)$ nomadic walk}\\ W'_0 = u,\\\text{$e^*$ s.t. $(e^*)_1=w_{k-2},(e^*)_2=v$}\\ \Atom(W'_{k-2})\ne\Atom(e^*)}} (\lambda_1+\lambda_2)\Wt(e^*)\Wt(W')\\
    &= (\lambda_1+\lambda_2)\sum_{W'~\text{length-$(k-1)$ nomadic walk from $u$ to $v$}} \Wt(W')\\
    &= (\lambda_1+\lambda_2)p_{k-1}(A)_{uv}.
\end{align*}

\noindent Now, we have for $k \ge 3$,
\begin{align*}
    \sum_{W\in\calW_{uv}} \Wt(W) &= \sum_{W\in\calW_{uv}^{(1)}} \Wt(W) + \sum_{W\in\calW_{uv}^{(2)}} \Wt(W) + \sum_{W\in\calW_{uv}^{(3)}} \Wt(W)\\
    Ap_{k-1}(A)_{uv} &= \sum_{W\in\calW_{uv}^{(1)}} \Wt(W) + (c-1)(-\lambda_1\lambda_2)p_{k-2}(A)_{uv} + (\lambda_1+\lambda_2)p_{k-1}(A)_{uv}\\
    \sum_{W\in\calW_{uv}^{(1)}} \Wt(W) &= Ap_{k-1}(A)_{uv} - ((c-1)(-\lambda_1\lambda_2)p_{k-2}(A)_{uv} + (\lambda_1+\lambda_2)p_{k-1}(A)_{uv})\\
    \sum_{W\in\calW_{uv}^{(1)}} \Wt(W) &= p_k(A)_{uv}.
\end{align*}
For the case of $k = 2$, we carry out the above calculation by replacing $(c-1)(-\lambda_1\lambda_2)$ with $c(-\lambda_1\lambda_2)$, thus completing the inductive step.
\end{proof}

\paragraph{Generic generating functions facts.}  Before returning to the specifics of our problem, we give some ``standard'' generating function facts.  These are extensions of the following simple idea: if $f(t)$ is a polynomial, then $\frac{d}{dt} \log f(t) = f'(t) \cdot f(t)^{-1}$ is (up to minor manipulations) the generating function for the power sum polynomials of its roots.  We start with a general matrix version of this, which is sometimes called \emph{Jacobi's formula} (after minor manipulations): %\footnote{See, e.g., \url{https://mathoverflow.net/a/127152/658}, \url{https://mathoverflow.net/a/214924/658}, \url{https://folk.ntnu.no/hanche/notes/diffdet/diffdet.pdf}}
\begin{proposition}                                     \label{prop:jacobi}
    Let $M(t)$ be a square matrix polynomial of~$t$.  Then
    \[
        \frac{d}{dt} \log \det M(t) = \tr\parens*{M'(t) M(t)^{-1}}
    \]
    for all $t\in\R$ such that $M(t)$ is invertible.
\end{proposition}
\begin{corollary}                                       \label{cor:gen-eigs}
    Taking $M(t) = \Id - H t$ for a fixed square matrix~$H$ yields
    \[
        \frac{d}{dt} \log \det (\Id - H t) = \tr\parens*{-H (\Id - Ht)^{-1}} \quad\implies\quad -t \frac{d}{dt} \log \det (\Id - H t) = \sum_{k \geq 1} \tr(H^k) t^k.
    \]
\end{corollary}
Regarding this corollary, we can derive the statement about the power sums of the roots of a polynomial~$f(t)$ by taking $H = \diag(\lambda_1, \dots, \lambda_n)$ where the $\lambda_i$'s are the roots of~$f$.  On the other hand, it actually suffices to prove \pref{cor:gen-eigs} in the case of diagonal~$H$, since $\det(\Id - Ht)$ is invariant to unitary conjugation.

\paragraph{Growth Rate.}  A key term that shows up in our Ihara--Bass formula is the ``growth rate'' of the additive product of $\calA$.  Suppose we take $t$-step nomadic walk starting at a vertex $v$ in the additive product graph, take a $t$-step nomadic walk back to $v$, and then sum over the total weight of such walks.  What we get is $\left((c-1)(-\lambda_1\lambda_2)\right)^t$ (see \pref{lem:growth-rate} for a proof).  Thus, the total weight of aforementioned walks grows exponentially in $t$ at a rate of $(c-1)(-\lambda_1\lambda_2)$, which in this section we will refer to as $\gr$.

\paragraph{The fundamental recurrence.}  We now relate the generating function matrix $F(t)$ to~$A$.  Using the recurrence used to generated the polynomials $p_k(x)$, one can conclude
\begin{lemma}                                       \label{lem:recur}
    $\displaystyle F(t) = A F(t)t - (\lambda_1+\lambda_2)F(t)t - \gr F(t) t^2 + (1+t\lambda_1)(1+t\lambda_2) \Id$.
\end{lemma}
From this recurrence one may express the inverse of $F(t)$ in terms of~$A$ and~$c$:
\begin{corollary}                                       \label{cor:recur}
    $\displaystyle (1+\lambda_1t)^{-1}(1+\lambda_2t)^{-1} \cdot (\Id - At + (\lambda_1+\lambda_2)t\Id + \gr t^2\Id)F(t) = \Id$.  In other words, $\displaystyle F(t) = (1+\lambda_1t)(1+\lambda_2t)\Id \cdot L(t)^{-1}$, where $L(t) \coloneqq \Id - At + (\lambda_1+\lambda_2)t\Id + \gr t^2\Id$ is the ``deformed Laplacian'' appearing in the statement of our Ihara--Bass theorem.
\end{corollary}

\paragraph{Strategy for the rest of the proof.} The strategy will be to apply \pref{prop:jacobi} with the deformed Laplacian~$L(t)$. On the left side we'll get a determinant involving~$A$.  On the right side we'll get a trace involving $L(t)^{-1}$, which is essentially $F(t)$.  In turn, $\tr(F(t))$ is a generating function for nomadic closed walks, which we can hope to relate to~$B$ (although there will be an edge case to deal with).

Let's begin executing this strategy.  By \pref{prop:jacobi} we have
\begin{align*}
    -t \frac{d}{dt} \log \det L(t) &= -t \cdot \tr\parens*{L'(t) L(t)^{-1}} \\
    &= -t \cdot \tr\parens*{(\Id(\lambda_1+\lambda_2) - A + 2\gr t\Id) \cdot ((1+\lambda_1t)(1+\lambda_2t))^{-1} F(t)}\\
    &= \frac{1}{(1+\lambda_1t)(1+\lambda_2t)} \tr\parens*{-(\lambda_1+\lambda_2)F(t)t + A F(t) t - 2\gr F(t) t^2 }
\end{align*}
where we used \pref{cor:recur}.  Now using \pref{lem:recur} again we may infer
\[
    -(\lambda_1+\lambda_2)F(t)t + A F(t) t - 2\gr F(t) t^2 = (1- \gr t^2) F(t) - (1+\lambda_1t)(1+\lambda_2t)\Id;
\]
combining the previous two identities yields
\begin{equation}    \label{eqn:halfway}
    -t \frac{d}{dt} \log \det L(t) = \tr\parens*{\frac{1-\gr t^2}{(1+\lambda_1t)(1+\lambda_2t)} F(t) - \Id}.
\end{equation}

\paragraph{Nomadic walks.}  The right side above is $\tr(F(t))$ up to some scaling/translating.  By definition, $\tr(F(t))$ is the generating function for nomadic \emph{circuits} (closed walks) with any starting point.  A first instinct is therefore to expect that
\begin{equation}    \label{eq:wrong}
    \tr(F(t)) \overset{?}{=} \sum_{k \geq 0} \tr(B^k) t^k,
\end{equation}
as $\tr(B^k)$ is the weight of closed length-$k$ circuits of direct edges in the nomadic world.  However this is not quite right: $\tr(B^k)$ only weighs the nomadic circuits whose first and last edge are not in the same atom.  The nomadic circuits that are not weighed can be identified either as (i) ``tailed'' nomadic circuits, i.e., those where the last directed edge is the reverse of the first directed edge;  (ii) ``stretched'' nomadic circuits, i.e., those where the last directed edge is distinct from but in the same atom as the first directed edge.  E.g., $\tr(B^k)$ would fail to count the following:

\myfig{.5}{example_tail}{A length-$9$ nomadic walk from $u$ to $u$ with a \emph{tail} of length $2$}{fig:tail}

Thus we need to correct \pref{eq:wrong}.  
\begin{definition}
    With the $-\Id$ taking care of the omission of $k = 0$, we define
    \begin{equation} \label{eqn:tails}
        \tails(t) = \sum_{k \geq 1} \parens*{\text{weight of nomadic circuits of length $k$}} t^k = \tr(F(t) - \Id).
    \end{equation}
    We also define
    \begin{align*}
        \notails(t) = \sum_{k \geq 1} \parens*{\text{weight of tail-less nomadic circuits of length $k$}} t^k
    \end{align*}
    and
    \begin{multline} \label{eqn:notails}
        \Simple(t) = \sum_{k \geq 1} \parens*{\text{weight of non-stretched, tail-less nomadic circuits of length $k$}} t^k \\
        = \sum_{k \geq 1} \tr(B^k) t^k = -t \frac{d}{dt} \log \det (\Id - Bt),
    \end{multline}
    where the last equality used \pref{cor:gen-eigs}.
\end{definition}

\paragraph{Tails vs.\ no tails vs.\ simple: more generating functions.}  We finish by relating $\tails(t)$, $\notails(t)$ and $\Simple(t)$.  This is the recipe:
\begin{quotation}
    A general nomadic circuit of length $k$ is constructed from a tail-less nomadic circuit of length $k-2\ell$ with a tail of length-$\ell$ attached to one of its vertices.
\end{quotation}
Tail-less nomadic circuits can be classified as (i) non-stretched tail-less nomadic circuits, and (ii)~stretched, tail-less nomadic circuits, for which,
\[
    \notails(t) - \Simple(t) = \sum_{k\ge 1}(\text{weight of stretched, tail-less nomadic walks of length $k$})t^k.
\]

Consider a stretched, tail-less nomadic walk of length $k$ that starts at vertex $v$, takes the edge $e$ from $v$ to $u$, goes on a nomadic walk $W$ from $u$ to $w$, and finally takes edge $e'$ from $w$ to $v$ to end the walk at $v$.  Note that $e$ and $e'$ are part of the same atom $A_i$.  Summing over all $v$ in atom $A_i$ and applying \pref{fact:off-diag-square} gives
\[
    \sum_{v\sim A_i} \Wt(A_i(v,u))\Wt(A_i(w,v))\Wt(W) = (\lambda_1+\lambda_2)\Wt(A_i(w,u))\Wt(W) = (\lambda_1+\lambda_2)\Wt(W')
\]
where $W'$ is a nomadic circuit of length $k-1$ that starts at $w$, takes edge $A_i(w,u)$ in the first step, and then takes walk $W$. From this, we derive
\[
    \notails(t) - \Simple(t) = (\lambda_1+\lambda_2)t\cdot\Simple(t).
\]

It's easy to count the total weight of tails of length $\ell$ one can attach to a given vertex of a tail-less nomadic circuit: if the tail-less nomadic circuit is non-stretched, the first edge can be chosen by picking any edge in $(c-2)$ atoms and each of the remaining $\ell-1$ edges can be chosen by picking any edge $(c-1)$ atoms; and if the tail-less nomadic circuit is stretched, each edge (including the first one) can be chosen anywhere from $(c-1)$ atoms.  From this it's easy to derive
\begin{multline}   \label{eqn:cf}
    \tails(t) = \parens*{1 + (-\lambda_1\lambda_2)(c-2)t^2  + (-\lambda_1\lambda_2)^2(c-2)(c-1) t^4 + \cdots} \Simple(t) \\
    + \parens*{1 + (-\lambda_1\lambda_2)(c-1)t^2 + (-\lambda_1\lambda_2)^2(c-1)^2t^4 + \cdots}(\notails(t)-\Simple(t)) \\
    = \frac{1-(-\lambda_1\lambda_2)t^2}{1-(c-1)(-\lambda_1\lambda_2)t^2} \Simple(t) + \frac{(\lambda_1+\lambda_2)t}{1-(c-1)(-\lambda_1\lambda_2)t^2}\Simple(t) \\
    \iff \Simple(t) =  \frac{1-\gr t^2}{(1+\lambda_1t)(1+\lambda_2t)} \tails(t).
\end{multline}

Using $\tails(t) = \tr(F(t) - \Id)$ (i.e., \pref{eqn:tails}), we obtain:
\begin{corollary}                                     \label{cor:tails-vs-notails}
    $\displaystyle  \Simple(t) = \tr\parens*{\frac{1-\gr t^2}{(1+\lambda_1t)(1+\lambda_2t)}(F(t) - \Id)}.$
\end{corollary}
But this is \emph{almost} the same as \pref{eqn:halfway}.  The difference is
\begin{align*}
    \tr\parens*{\Id - \frac{1-\gr t^2}{(1+\lambda_1t)(1+\lambda_2t)}\Id} &= \tr\parens*{\frac{(\lambda_1+\lambda_2)t+(c-2)(-\lambda_1\lambda_2)t^2}{(1+\lambda_1t)(1+\lambda_2t)}\Id}\\
    &= |V|\cdot\frac{(\lambda_1+\lambda_2)t+(c-2)(-\lambda_1\lambda_2)t^2}{(1+\lambda_1t)(1+\lambda_2t)}.
\end{align*}
Combining the above with \pref{eqn:halfway}, \pref{cor:tails-vs-notails}, and \pref{eqn:notails}, we finally conclude
\[
    -t \frac{d}{dt} \log \det L(t)  + |V|\cdot\frac{(\lambda_1+\lambda_2)t+(c-2)(-\lambda_1\lambda_2)t^2}{(1+\lambda_1t)(1+\lambda_2t)} =  -t \frac{d}{dt} \log \det (\Id - Bt).
\]
Finally, dividing by $-t$, integrating (which leaves an unspecified additive constant), and exponentiating (now there is an unspecified multiplicative constant) yields
\[
     \text{(const.)} \cdot (1+\lambda_1t)^{|V|\frac{ c\lambda_2}{\lambda_2-\lambda_1}-1}(1+\lambda_2t)^{|V|\frac{c\lambda_1}{\lambda_1-\lambda_2}-1} \det L(t) = \det(\Id - Bt).
\]
By consideration of $t = 0$ we see that the constant must be $1$.

%% file: bound.tex
\section{Connecting the adjacency and nomadic spectrum} \label{sec:bound}

Let $\atomlist = (A_1,\dots,A_c)$ be a sequence of atoms with two distinct eigenvalues $\lambda_1$ and $\lambda_2$, let $\congraph$ be an $r$-ary, $c$-atom constraint graph, and let $\instgraph=\atomlist(\congraph)$ be the corresponding instance graph.  We use $A$ for the adjacency matrix of $\instgraph$, $B$ for its nomadic walk matrix, $V$ for its vertex set, and $E$ for its edge set.  Recall that $\gr$ is defined as $(c-1)(-\lambda_1\lambda_2)$.

We want to use \pref{thm:ihara-bass-restatement} to describe the spectrum of $B$ with respect to that of $A$. We will refer to eigenvalues of $B$ with the letter $\mu$ and eigenvalues of $A$ with the letter $\nu$.

First, notice that if $t$ is such that $\det(\Id - Bt) = 0$, then  $\mu = 1 / t$ has $\det(\mu \Id - B) = 0$, meaning $\mu$ is an eigenvalue of $B$. Thus we want to find for which values of~$t$ does the left-hand side of the expression in \pref{thm:ihara-bass} become $0$ in order to deduce the spectrum of $B$.

It is easy to see that when $t = - 1 / \lambda_1$ and $t = - 1 / \lambda_2$ the left-hand side is always $0$, so $-\lambda_1$ is an eigenvalue of $B$ with multiplicity $|V|(\frac{c \lambda_2}{\lambda_2 - \lambda_1} - 1)$ and $-\lambda_2$ is an eigenvalue with multiplicity $|V|(\frac{c \lambda_1}{\lambda_1 - \lambda_2} - 1)$. The remaining eigenvalues are given by the values of $t$ for which $\det(L(t)) = 0$. Let $t$ be such that $\det(L(t)) = 0$; then we have that $L(t)$ is non-invertible, which means there is some vector $v$ in the nullspace of $L(t)$. By rearranging the equality $L(t) v = 0$ we get:

$$
  Av = \frac{1 + (\lambda_1 + \lambda_2)t + \gr t^2}{t} v.
$$

This implies that $\frac{1 + (\lambda_1 + \lambda_2)t + \gr t^2}{t}$ is an eigenvalue of $A$. Let $\nu$ be some eigenvalue of $A$; then we have that $\nu = \frac{1 + (\lambda_1 + \lambda_2)t + \gr t^2}{t}$ for some $t$. If we rearrange the previous expression we get the following quadratic equation in $t$:

$$
  1 + (\lambda_1 + \lambda_2 - \nu)t + \gr t^2 = 0.
$$

By solving this expression for $t$ and then using the fact that $\mu = 1 / t$ we get (notice that $c > 1$ and $\lambda_1\lambda_2 \neq 0$):

$$
  \mu = \frac{-2\gr}{\lambda_1 + \lambda_2 - \nu \pm \sqrt{(\lambda_1 + \lambda_2 - \nu)^2 - 4\gr}}.
$$

To analyze the previous we look at three cases:

\begin{enumerate}
    \item $\nu > \lambda_1 + \lambda_2 + 2\sqrt{\gr}$. In this case the discriminant is always positive. If we look at the $-$ branch of the $\pm$ we further get that the denominator of the previous formula is always less than $-2\sqrt{\gr}$ which means we have that $\mu$ is real and $\mu > \sqrt{\gr}$. Additionally, we have that in this interval $\mu$ is an increasing function of $\nu$.
    \item $\nu < \lambda_1 + \lambda_2 - 2\sqrt{\gr}$. This is analogous to the previous case; if we look at the $+$ branch we have that $\mu$ is real and $\mu < -\sqrt{\gr}$. Additionally, we have that in this interval $\mu$ is a decreasing function of $\nu$.
    \item $\nu \in [\lambda_1 + \lambda_2 - 2\sqrt{\gr}, \lambda_1 + \lambda_2 + 2\sqrt{\gr}]$, for each such $\nu$ we get a pair of anti-conjugate complex numbers, meaning a pair $x, \bar{x}$ such that $x \bar{x} = -1$.
\end{enumerate}

Finally, the spectrum of $B$ also contains 0 with multiplicity $2|E| - |V|\left(2 + (\frac{c \lambda_1}{\lambda_1 - \lambda_2} - 1) + (\frac{c \lambda_2}{\lambda_2 - \lambda_1} - 1) \right)$, which we get because the degrees of the polynomials in the left-hand side and right-hand do not match; the right-hand side has degree $2|E|$ but we only described $|V|\left(2 + (\frac{c \lambda_1}{\lambda_1 - \lambda_2} - 1) + (\frac{c \lambda_2}{\lambda_2 - \lambda_1} - 1) \right)$ roots.

We can now summarize the eigenvalues of $B$ in the following way:

\begin{itemize}
    \item $-\lambda_1$ with multiplicity $|V|(\frac{c \lambda_2}{\lambda_2 - \lambda_1} - 1)$;
    \item $-\lambda_2$ with multiplicity $|V|(\frac{c \lambda_1}{\lambda_1 - \lambda_2} - 1)$;
    \item for each eigenvalue $\nu$ of $A$ we get two eigenvalues that are solutions to the previous quadratic equation;
    \item 0 with multiplicity $2|E| - |V|\left(2 + (\frac{c \lambda_1}{\lambda_1 - \lambda_2} - 1) + (\frac{c \lambda_2}{\lambda_2 - \lambda_1} - 1) \right)$;
\end{itemize}

The distribution of the eigenvalues that come from $A$ forms a sort of semicircle. To showcase this behavior we display an example of the spectrum of typical lifted instance in \pref{fig:specB}.

\begin{figure}[!ht]
  \centering
  \includegraphics[width=0.5\textwidth]{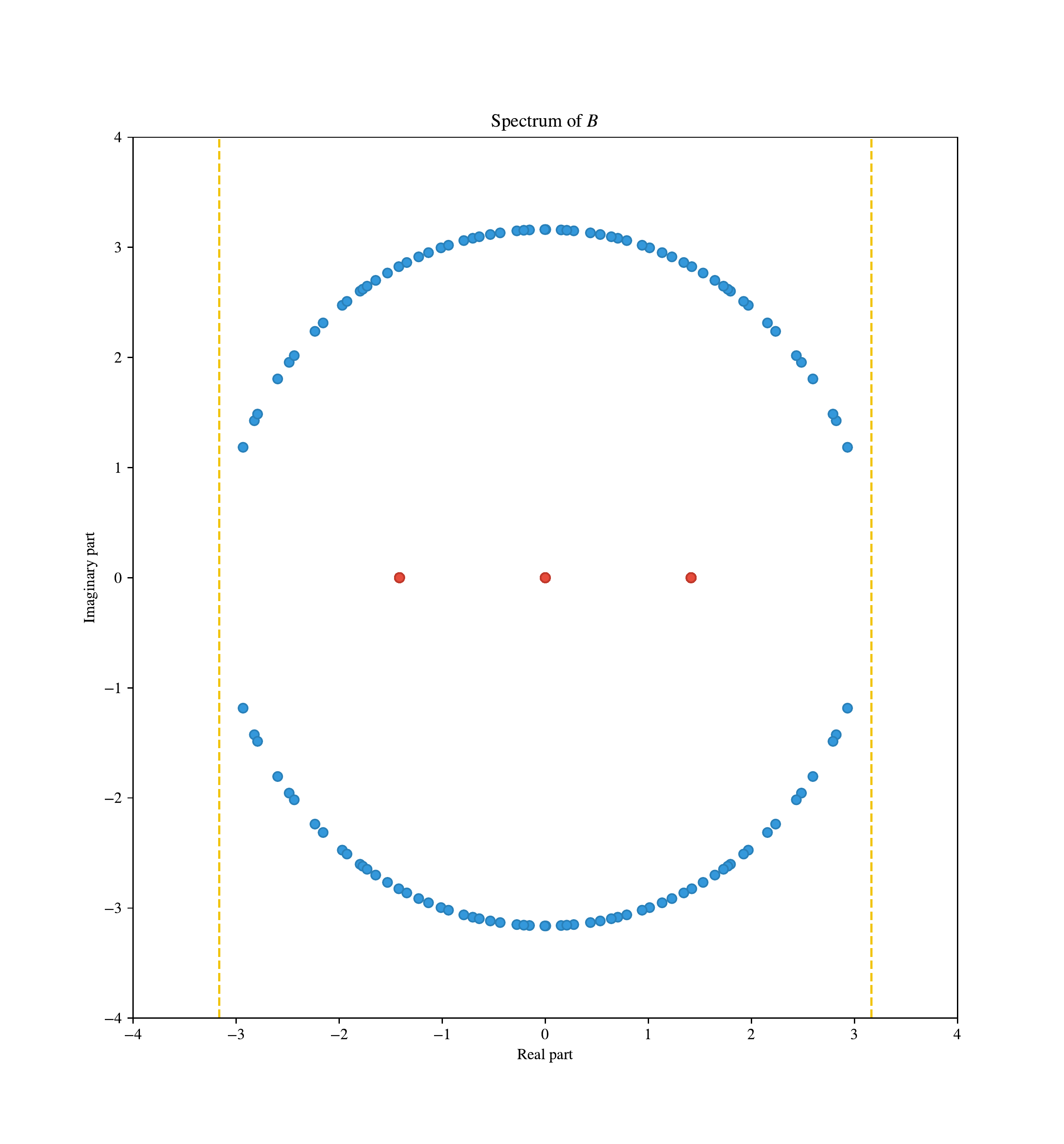}

  \caption{The spectrum of $B$ for a additive $15$-lift of $6$ copies of a $\sort_4$ graph. The blue dots are eigenvalues that come from eigenvalues of $A$, the red dots are either $-\lambda_1$, $-\lambda_2$ or 0 and the yellow line is the limit $\sqrt{\gr}$.}
  \label{fig:specB}
\end{figure}

\begin{figure}[!ht]
  \centering
  \begin{subfigure}[b]{0.45\textwidth}
    \centering
    \includegraphics[width=\textwidth]{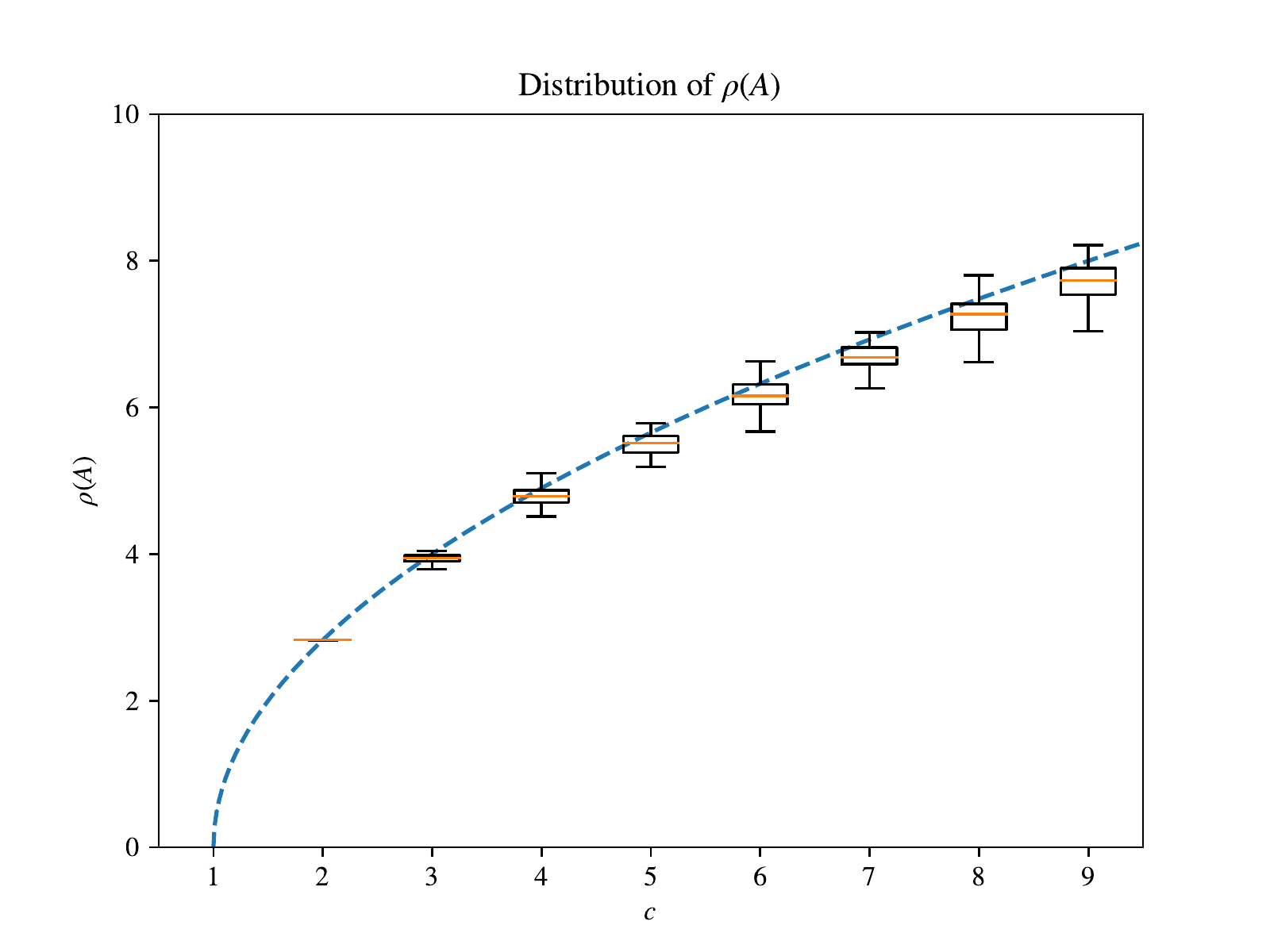}
  \end{subfigure}
  ~
  \begin{subfigure}[b]{0.45\textwidth}
    \centering
    \includegraphics[width=\textwidth]{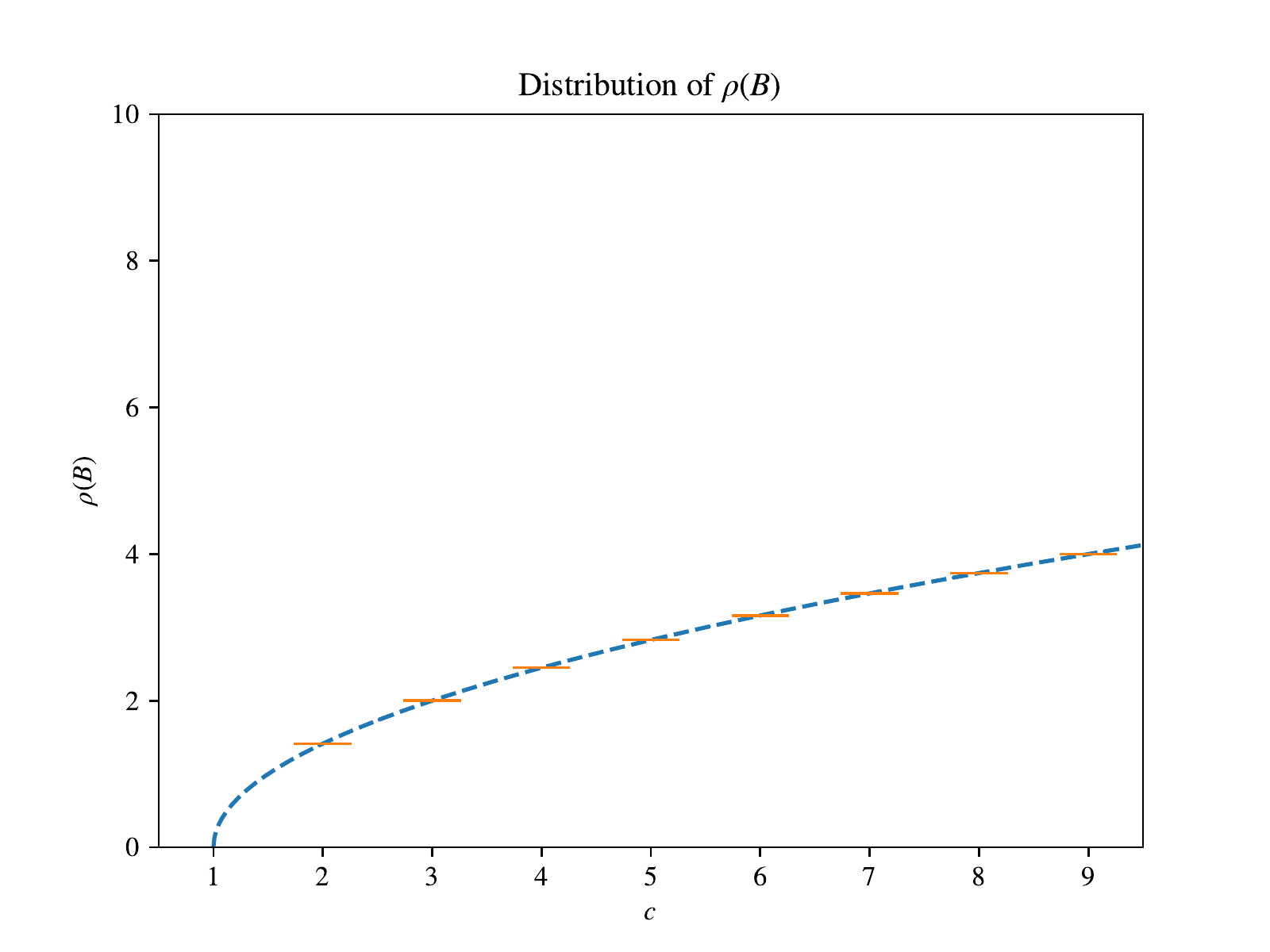}
  \end{subfigure}
  
  \caption{A box plot of $\rho(A)$ and $\rho(B)$ of 100 samples of random instance graphs as a function of $c$ with $n = 15$, $r = 4$ and all atoms are the $\sort_4$ graph. The dashed line shows the theoretical bound prediction of $2\sqrt{\gr}$ for $A$ and $\sqrt{\gr}$ for $B$.}
  \label{fig:practice}
\end{figure}

We can now prove the central theorem of this section:

\begin{theorem}
  \label{thm:boundadj}
    Let $\binstgraph_n$ be a random additive $n$-lift of $\atomlist$ with adjacency matrix $A_{\binstgraph_n}$, and let $\epsilon > 0$. Then:

    $$\mathbf{Pr} \left [ \rho(A_{\binstgraph_n}) \in [\lambda_1 + \lambda_2 - 2\sqrt{\gr} - \eps, \lambda_1 + \lambda_2 + 2\sqrt{\gr} + \eps \right ] = 1 - o_{n}(1)$$
\end{theorem}
\begin{proof}
  First recall \pref{thm:adjacency-spec-bounds-informal} (for fully formal statement, see \pref{thm:C}) and notice that $\rho(|B|) = \gr$, which follows by using the trivial upper bound of $\gr^{2k}$ on $\tr\left(|B|^k\left(|B|^*\right)^k\right)$.  From cases 1 and 2 in the previous analysis we get that if $\rho(A_{\binstgraph_n}) \notin [\lambda_1 + \lambda_2 - 2\sqrt{\gr} - \eps, \lambda_1 + \lambda_2 + 2\sqrt{\gr} + \eps]$ there is some constant $\delta$ such that $\rho(B_n) > \sqrt{\gr} + \delta$, which happens with $o_{n \to \infty}(1)$ probability by \pref{thm:C}.
\end{proof}

Also, we note that even though throughout our proof we hide various constant factors, the bounds obtained in \pref{thm:boundadj} and \pref{thm:C} are empirically visible for very small values of $n$ and $c$. To justify this claim we show in \pref{fig:practice} a plot of samples of random instance graphs for different values of $c$ with a fixed small $n$.

%% file: gaussian-wave.tex
\section{Additive products of 2-eigenvalue atoms} \label{sec:wave}

In this section, we let $\calA = (A_1,\dots,A_c)$ be a sequence of $\{\pm 1\}$-\emph{weighted} atoms with the same pair of exactly two distinct eigenvalues, $\lambda_1$ and $\lambda_2$.  We also let $X \coloneqq A_1\addprod\cdots\addprod A_c$ be the additive product graph.  We use $A_X$ to denote the adjacency operator of $X$.  In this section, $\binstgraph_n$ is the instance graph of a random additive $n$-lift of $\calA$ with negations, and we use $A_{\binstgraph_n}$ to denote its adjacency matrix.  Finally, we recall $\gr:=(c-1)(-\lambda_1\lambda_2)$ and define the quantity $r_X := 2\sqrt{\gr}$.

The main results that this section is dedicated to proving are:
\begin{theorem} \label{thm:spectrum-X}
    The following are true about the spectrum of $X$:
    \begin{enumerate}
        \item \label{item:containment} $\spec(A_X)\subseteq[\lambda_1+\lambda_2-r_X,\lambda_1+\lambda_2+r_X]$;
        \item \label{item:hull} $\lambda_1+\lambda_2 - r_X$ and $\lambda_1+\lambda_2+r_X$ are both in $\spec(A_X)$.
    \end{enumerate}
\end{theorem}

\begin{theorem} \label{thm:SDP-sol}
    For every $\eps>0$, for large enough $n$, there are $|V(\binstgraph_n)|\times |V(\binstgraph_n)|$ positive semidefinite matrices $M_+$ and $M_-$ with all-ones diagonals such that
    \begin{align*}
        \langle A_{\binstgraph_n}, M_+\rangle&\ge(\lambda_1+\lambda_2+r_X-\eps)n\\
        \langle A_{\binstgraph_n}, M_-\rangle&\le(\lambda_1+\lambda_2-r_X+\eps)n.
    \end{align*}
    with probability $1-o_n(1)$.
\end{theorem}
In this section, when we measure the distance between vertices $u$ and $v$ in an instance graph $\binstgraph_n$, we look at the corresponding vertices in the constraint graph $\bcongraph$, and define $d(u,v) := \frac{d_{\bcalK}(u,v)}{2}$.  We use $\calP_{uv}$ to refer to the collection of edges comprising the shortest path between $u$ and $v$.  We begin with a statement about the `growth rate' of $X$.
\begin{lemma}   \label{lem:growth-rate}
    For all vertices $v$ in $V(X)$, for $t\ge 1$ we have
    \[
        \sum_{u:d(u,v)=t} \prod_{\{i,j\}\in\calP_{uv}} (A_X)_{ij}^2 = c(c-1)^{t-1}(-\lambda_1\lambda_2)^t.
    \]
\end{lemma}
\begin{proof}
    We proceed by induction.  When $t = 1$, the statement immediately follows from \pref{fact:diag-square}.  Suppose the equality is true for some $t=\ell-1$, we will show how statement follows for $t=\ell$. 
    \begin{align*}
        \sum_{u:d(u,v)=\ell} \prod_{\{i,j\}\in\calP_{uv}} (A_X)_{ij}^2 &= \sum_{u:d(u,v)=\ell-1} \left(\prod_{\{i,j\}\in \calP_{uv}}(A_X)_{ij}^2\right)\cdot\left(\sum_{\substack{u'\in N(u)\\ d(u',v)=\ell}}(A_X)_{uu'}^2\right)
    \intertext{From \pref{fact:diag-square}, $\displaystyle \sum_{\substack{u'\sim u\\ d(u',v)=t}}(A_X)_{uu'}^2$ is equal to $(c-1)(-\lambda_1\lambda_2)$, which means the above is equal to}
    &= \sum_{u:d(u,v)=\ell-1} \left(\prod_{\{i,j\}\in \calP_{uv}}(A_X)_{ij}^2\right)(c-1)(-\lambda_1\lambda_2)\\
    &= (c-1)^{\ell-2}c(-\lambda_1\lambda_2)^{\ell-1}(c-1)(-\lambda_1\lambda_2)\\
    &= c(c-1)^{\ell-1}(-\lambda_1\lambda_2)^{\ell}. \qedhere
    \end{align*}
\end{proof}

\begin{corollary}   \label{cor:deg-X}
    Since all the weights of $X$ are $\{\pm 1\}$-valued, the degree of every vertex in $X$ equals $c(-\lambda_1\lambda_2)$.
\end{corollary}

\subsection{Enclosing the spectrum}
Let $B_X$ denote the nomadic walk operator of $X$.  In this section, we show
\[
    \spec(A_X)\subseteq\left[\lambda_1+\lambda_2-r_X,\lambda_1+\lambda_2+r_X\right].
\]
The first part of the proof will involve showing that the spectral radius of $B_X$ is bounded by $\sqrt{\gr}$, and the second part translates this bound to the desired one on $\spec(A_X)$.  Both these components closely follow proofs from the work of Angel et al.; the former after \cite[Theorem 4.2]{AFH15} and the latter after \cite[Theorem 1.5]{AFH15}.
\begin{lemma}   \label{lem:infinite-nomadic-bound}
    $\spec(B_X)\subseteq \left[-\sqrt{\gr}, \sqrt{\gr}\right]$.
\end{lemma}
\begin{proof}
    Arbitrarily fix a root $r$ of $X$.  Recall that the spectral radius of $B_X$ is equal to $\lim\left(\|B_X^k\|_{\op}\right)^{1/k}$, and hence it suffices to bound $\left|\langle g, B_X^k f \rangle\right|$ for arbitrary $f$ and $g$ with $\|f\|=\|g\|=1$.

    We can decompose every nomadic walk of length $k$ into two segments, a segment of $i$ steps towards $r$ followed by a sequence of $k-i$ steps away from $r$; henceforth, we call length-$k$ nomadic walks with such a decomposition $(i,k)$-nomadic walks.  For every pair of directed edges $e$ and $e'$ such that $e,e_1,\dots,e_{k-1},e'$ is an $(i,k)$-nomadic walk, let $a(e,e') := \gr^{k/2-i}$.  From \pref{lem:growth-rate}, the number of $(i,k)$-nomadic walks starting at a fixed $e$ is at most $\frac{c}{c-1}\gr^{k-i}$.  Similarly, the number of $(i,k)$-nomadic walks ending at fixed $e'$ is at most $\frac{c}{c-1}\gr^{i}$.  Now, we are ready to bound $\left|\langle g, B_X^k f \rangle\right|$ by imitating the proof of \cite[Theorem 4.2]{AFH15}.
    \begin{align*}
        \left|\langle g, B_X^k f \rangle\right| &\le \left|\sum_{e,e_1,\dots,e_{k-1},e'~\text{nomadic}}f(e')g(e)\right|\\
        &\le \sum_{e,e_1,\dots,e_{k-1},e'~\text{nomadic}} |f(e')g(e)|\\
        &\le \sum_{e,e_1,\dots,e_{k-1},e'~\text{nomadic}} a(e,e')f(e')^2+\frac{1}{a(e,e')}g(e)^2\\
        &\le \sup_{e'}\left(\sum_{e,e_1,\dots,e_{k-1},e'~\text{nomadic}}a(e,e')\right)\|f\|_2^2 + \sup_e\left(\sum_{e,e_1,\dots,e_{k-1},e'~\text{nomadic}}\frac{1}{a(e,e')}\right)\|g\|_2^2\\
        &\le \sum_{i=0}^k\sup_{e'}\left(\sum_{(i,k)\text{-nomadic walks ending at $e'$}}a(e,e')\right) + \sup_e \left(\sum_{(i,k)\text{-nomadic walks starting at $e$}}\frac{1}{a(e,e')}\right)\\
        &\le \sum_{i=0}^k \gr^{k/2-i}\cdot\frac{c}{c-1}\gr^i + \sum_{i=0}^k \gr^{i-k/2}\cdot\frac{c}{c-1}\gr^{k-i}\\
        &= \frac{2kc}{c-1}\gr^{k/2}
    \end{align*}
    Thus, we have
    \[
        \|B_X^k\|_{\mathrm{op}} \le \frac{2kc}{c-1}\gr^{k/2}
    \]
    and taking the limit of $\|B_X^k\|_{\mathrm{op}}^{1/k}$ for $k$ approaching infinity yields the desired statement.
\end{proof}

\begin{lemma}   \label{lem:one-sided-IB}
    If $0$ is an approximate eigenvalue of $Q_t \coloneqq (t^2+(c-1)(-\lambda_1\lambda_2))\Id-A_Xt+(\lambda_1 + \lambda_2)\Id t$, then it is also an approximate eigenvalue of $B_X-t\Id$ as long as $t \ne -\lambda_1,-\lambda_2$.
\end{lemma}
\begin{proof}
    Let $f$ be an $\eps$-approximate eigenfunction of unit norm of $Q_t$, then we construct a $C\eps$-approximate eigenfunction $g$ of $B_X-t\Id$ defined on pairs $uv$ such that $u$ and $v$ are incident to a common atom for an absolute constant $C > 0$ as follows,
    \[
        g_{uv} := \left(\sum_{w:\{v,w\}\in\Atom(\{u,v\})} (A_X)_{vw}f_w\right) - (\lambda_1+\lambda_2+t)f_v
    \]
    for every edge $\{u,v\}$ of $X$.
    \begin{align*}
        \left((B_X-t\Id)g\right)_{uv} &= \left(\sum_{\substack{w:\\ \{v,w\}\notin\Atom(\{u,v\})}} (B_X)_{uv,vw} g_{vw}\right) - t g_{uv}\\
        &= \left(\sum_{\substack{w:\\ \{v,w\}\notin\Atom(\{u,v)\}}} (A_X)_{vw}\left( \sum_{\substack{x:\\ \{w,x\}\in\Atom(\{v,w\})}}(A_X)_{wx}f_x - (\lambda_1+\lambda_2+t)f_w\right)\right) - t g_{uv}\\
        &= \left(\sum_{\substack{w:\\ \{v,w\}\notin\Atom(\{u,v)\}}} \sum_{\substack{x:\\ \{w,x\}\in\Atom(\{v,w\})}} (A_X)_{vw} (A_X)_{wx}f_x\right) -\\ &~~~~~\left(\sum_{\substack{w:\\ \{v,w\}\notin\Atom(\{u,v)\}}} (\lambda_1+\lambda_2+t)(A_X)_{vw}f_w\right) - t g_{uv}
    \intertext{Using \pref{fact:diag-square} and \pref{fact:off-diag-square}, the first term of the three above can be rewritten as}
        &(c-1)(-\lambda_1\lambda_2)f_v+(\lambda_1+\lambda_2)\sum_{w:\{v,w\}\notin\Atom(\{u,v\})} (A_X)_{vw}f_w\\
    \intertext{which lets us continue the chain of equalities}
    &= (c-1)(-\lambda_1\lambda_2)f_v-t\sum_{\substack{w:\\ \{v,w\}\notin\Atom(\{u,v)\}}}(A_X)_{vw}f_w\\
    &- t\left(\sum_{w:\{v,w\}\in\Atom(\{u,v\})} (A_X)_{vw}f_w\right)+t(\lambda_1+\lambda_2+t)f_v\\
    &= (c-1)(-\lambda_1\lambda_2)f_v-t(Af)_v+t(\lambda_1+\lambda_2+t)f_v\\
    &= (Q_tf)_v.
    \end{align*}
    Thus,
    \[
        \|(B_X-t\Id)g\|_2^2 = \sum_{\{u,v\}\in E(X)} ((B_X-t\Id)g)_{uv}^2 + ((B_X-t\Id)g)_{vu}^2 = d\sum_{v\in V}(Q_tf)_v^2 \le d\eps^2
    \]
    It remains to show that the norm of $g$ is bounded from above and below.  Fix a vertex $u$ and an atom $\wt{A}$ incident to $u$.  Consider $g^{(u,\wt{A})}$, the restriction of $g$ to entries ${uv}$ such that the edge $\{u,v\}$ is in $\wt{A}$, and $f^{(\wt{A})}$, the restriction of $f$ to vertices $v$ such that $\wt{A}$ is incident to $v$.  Observe that $g^{(u,\wt{A})} = (A_{\wt{A}}-(\lambda_1+\lambda_2+t)\Id)f^{(\wt{A})}$.  Since the min eigenvalue of $A_{\wt{A}}-(\lambda_1+\lambda_2+t)\Id$ is nonzero as long as $t\ne-\lambda_1,-\lambda_2$, the $\ell_2$ norm of $g$ is bounded from below.  To prove that the $\ell_2$ norm of $g$ is bounded from above, observe that
    \begin{align*}
        \|g\|_2^2 &= \sum_{\wt{A}\in\Atoms(X)} \sum_{(u,v):\{u,v\}\in\wt{A}}\left(\left(\sum_{w:\{v,w\}\in\wt{A}} (A_X)_{vw}f_w\right)-(\lambda_1+\lambda_2+t)f_v\right)^2\\
        &\le 2\sum_{\wt{A}\in\Atoms(X)}\sum_{(u,v):\{u,v\}\in\wt{A}} \left(\sum_{\{v,w\}\in\wt{A}}(A_X)_{vw}^2f_w^2 + (\lambda_1+\lambda_2+t)^2f_v^2\right)\\
    \end{align*}
    There is some coefficient $\alpha$ such that the weight on $f_v^2$ for each $v$ in the above sum is bounded by $\alpha$, thereby giving a bound of
    \[
        2\sum_{v\in V}\alpha f_v^2\le 2\alpha\|f\|_2^2 \le 2\alpha. \qedhere
    \]
\end{proof}

\begin{proof}[Proof of \pref{item:containment} in \pref{thm:spectrum-X}]
    Let $Q_t$ be as defined in the statement of \pref{lem:one-sided-IB}.  It can be verified that $0$ is an approximate eigenvalue of either $Q_{-\lambda_1}$ or $Q_{-\lambda_2}$ if and only if $d_X := c(-\lambda_1\lambda_2)$, which we recall from \pref{cor:deg-X} is the degree of every vertex in $X$, is in the spectrum of $A_X$.  Let $\mu_+ := \lambda_1+\lambda_2+r_X+\eta$ be in spectrum of $A_X$.  If $\mu_+ \ne d_X$, then we can conclude from \pref{lem:one-sided-IB} that
    \[
        \gr + \eta + \sqrt{\eta\gr+\eta^2/4}
    \]
    is an approximate eigenvalue of $B_X$.  Since $\spec(B_X)$ is contained in $[-\sqrt{\gr},\sqrt{gr}]$, $\eta$ cannot be positive.  A similar argument applied to $\mu_- := \lambda_1+\lambda_2-r_X-\eta$ precludes $\eta$ from being positive as long as $\mu_-\ne d_X$.  As a result, we can conclude that $\spec(A_X)$ is contained in $[\mu_-,\mu_+]\cup\{d_X\}$.  If $d_X$ is in the interval $[\mu_-,\mu_+]$, then we are done.  If not, then it remains to show that $d_X$ is not in $\spec(A_X)$.  Since $X$ is $\{\pm 1\}$-weighted and the degree of each vertex is $d_X$, any nonzero $x$ satisfying $A_Xx = d_Xx$ must have the same nonzero magnitude in all its entries.  However, such $x$ has unbounded $\ell_2$ norm, and hence $A_X$ has no eigenvectors with eigenvalue $d_X$ in $\ell_2(V)$.  If $d_X$ is in $\spec(A_X)$, it is an isolated point in the spectrum, and hence, by \pref{thm:isolation}, is an eigenvalue of $A_X$, which means $d_X$ cannot be in $\spec(A_X)$.
\end{proof}

\subsection{Construction of Witness Vectors}

\begin{lemma}[\pref{item:hull} of \pref{thm:spectrum-X} restated] \label{lem:witness-vectors}
    There exists $\lambda_- \le \lambda_1+\lambda_2-r_X$ and $\lambda_+ \ge \lambda_1+\lambda_2+r_X$ in the spectrum of $A_X$.
\end{lemma}
\begin{proof}
    Let $\delta > 0$ be a parameter to be chosen later.  First define $\rho$ as
    \[
        \rho(s) := \frac{s(1-\delta)}{\sqrt{(c-1)(-\lambda_1\lambda_2)}}
    \]
    Then, for vertex $v$ and define $f_v^{(s)}$ in the following way.
    \begin{align}
        f_v^{(s)}(u) := \rho(s)^{d(u,v)}\prod\limits_{\{i,j\}\in\calP_{uv}}\left(A_X\right)_{ij}~\text{ where $\calP_{uv}$ is the unique \permissible~walk between $u$ and $v$} \label{eq:}
    \end{align}
    To show the lemma, it suffices to prove the claim that for every $\eps > 0$, there is suitable choice of $\delta$ so that
    \[
        \frac{\langle f_v^{(-1)}, A_Xf_v^{(-1)}\rangle}{\langle f_v^{(-1)}, f_v^{(-1)}\rangle} < \lambda_1+\lambda_2-r_X+\eps
    \]
    and
    \[
        \frac{\langle f_v^{(1)}, A_Xf_v^{(1)}\rangle}{\langle f_v^{(1)}, f_v^{(1)}\rangle}>\lambda_1+\lambda_2+r_X-\eps
    \]

    We proceed by analyzing the expression $\langle f_v^{(s)}, A_X f_v^{(s)}\rangle$.
    \begin{align}
        \langle f_v^{(s)}, A_Xf_v^{(s)}\rangle &= \sum_{u\in V} f_v^{(s)}(u)A_Xf_v^{(s)}(u) \nonumber\\
        &= f^{(s)}_v(v)\sum_{w\in N(v)}\left(A_X\right)_{vw}f^{(s)}_v(w) + \sum_{u\in V,u\ne v} f_v^{(s)}(u)\sum_{w\in N(u)}(A_X)_{uw}f_v^{(s)}(w) \nonumber \\
        &= \sum_{w\in N(v)}(A_X)_{vw}^2\rho(s) + \sum_{u\in V,u\ne v} f_v^{(s)}(u)\sum_{w\in N(u)}(A_X)_{uw}f_v^{(s)}(w) \label{eq:rayleigh}
    \end{align}
    Let $w_0,w_1,\dots w_{T-1},w_T$ be the sequence of vertices from the unique \permissible walk between $u$ and $v$ where $w_0 = u$ and $w_T = v$.  Now, let $u^* = w_{1}$.  Recall the notation $\calP_{u,v}$ used to denote the unique \permissible walk between $u$ and $v$ as a sequence of edges.  Let $W_{u,v} := \rho(s)^{d(u,v)}\prod\limits_{\{i,j\}\in\calP_{u,v}}(A_X)_{ij}$.  Using the notation we just developed, along with applying \pref{fact:diag-square} on the first term of the above, we get
    \begin{align*}
        \pref{eq:rayleigh} &= c(-\lambda_1\lambda_2)\rho(s) + \sum_{u\in V, u\ne v} \rho(s)W_{u^*v}(A_X)_{uu^*}\cdot\\
        &\left((A_X)_{uu^*}W_{u^*v} + \sum_{w\in\mathrm{Atom}(\{u^*,u\})}\rho(s)(A_X)_{u^*w}(A_X)_{wu}W_{u^*v}+\sum_{\substack{w\notin\mathrm{Atom}(\{u,u^*\}) \\ w\in N(u)}}\rho(s)^2(A_X)_{u^*u}(A_X)_{uw}^2W_{u^*v}\right)\\
        &= c(-\lambda_1\lambda_2)\rho(s) + \sum_{u\in V, u\ne v} \rho(s)W_{u^*v}^2(A_X)_{uu^*}^2\cdot\\
        &\left(1 + \frac{\sum\limits_{w\in\mathrm{Atom}(\{u^*,u\})}\rho(s)(A_X)_{u^*w}(A_X)_{wu}}{A_{uu^*}} + \sum_{\substack{w\notin\mathrm{Atom}(\{u,u^*\}) \\ w\in N(u)}} (A_X)_{uw}^2\rho(s)^2\right)
        \intertext{Now we apply \pref{fact:diag-square} and \pref{fact:off-diag-square} and get}
        &= c(-\lambda_1\lambda_2)\rho(s) + \sum_{u\in V, u\ne v} \rho(s)W_{u^*v}^2(A_X)_{uu^*}^2\cdot \left(1 + \rho(s)(\lambda_1+\lambda_2) + (c-1)(-\lambda_1\lambda_2)\rho(s)^2\right)\\
        &= c(-\lambda_1\lambda_2)\rho(s) + \sum_{u\in V, u\ne v} W_{uv}^2\cdot\frac{1+\rho(s)(\lambda_1+\lambda_2)+(c-1)(-\lambda_1\lambda_2)\rho(s)^2}{\rho(s)}\\
        &= c(-\lambda_1\lambda_2)\rho(s) + \left(\|f_v^{(s)}\|^2-1\right)\cdot \frac{1+\rho(s)(\lambda_1+\lambda_2)+(c-1)(-\lambda_1\lambda_2)\rho(s)^2}{\rho(s)}\\
        &= c(-\lambda_1\lambda_2)\rho(s) + \left(\|f_v^{(s)}\|^2-1\right)\cdot \left(\frac{1+s^2(1-\delta)^2}{\rho(s)} + (\lambda_1+\lambda_2)\right)
    \end{align*}
    When $s = \pm 1$, the above quantity is equal to
    \begin{align*}
        c(-\lambda_1\lambda_2)\rho(s) + \left(\|f_v^{(s)}\|^2-1\right)\cdot \left(\frac{1+(1-\delta)^2}{\rho(s)} + (\lambda_1+\lambda_2)\right)
    \end{align*}

    Now, note that
    \begin{align}
        \frac{\langle f_v^{(s)}, A_Xf_v^{(s)}\rangle}{\langle f_v^{(s)}, f_v^{(s)}\rangle} &= \frac{c(-\lambda_1\lambda_2)\rho(s)}{\|f_v^{(s)}\|^2} + \left(1-\frac{1}{\|f_v^{(s)}\|^2}\right)\cdot\left(\frac{1+(1-\delta)^2}{\rho(s)} + (\lambda_1+\lambda_2)\right) \label{eq:other-rayleigh}
    \end{align}

    We now compute $\|f_v^{(s)}\|^2$, and we assume $s$ is either $+1$ or $-1$.
    \begin{align*}
        \|f_v^{(s)}\|^2 &= \sum_{t=0}^{\infty} \rho(s)^{2t}\sum_{u:d(u,v)=t} \prod_{\{i,j\}\in\calP_{uv}} (A_X)_{ij}^2\\
        &= \sum_{t=0}^{\infty}  \rho(s)^{2t} c(c-1)^{t-1} (-\lambda_1\lambda_2)^t &\text{(by \pref{lem:growth-rate})}\\
        &= \frac{c}{c-1}\sum_{t=0}^\infty \left(\frac{(1-\delta)^{2t}}{(c-1)^t(-\lambda_1\lambda_2)^t}\right) (c-1)^t(-\lambda_1\lambda_2)^t\\
        &= \frac{c}{c-1}\sum_{t=0}^{\infty}(1-\delta)^{2t}\\
        &= \frac{c}{c-1}\cdot\frac{1}{\delta(2-\delta)}
    \end{align*}

    Plugging this back in to \pref{eq:other-rayleigh} gives
    \begin{align*}
        \pref{eq:other-rayleigh} &= \delta(2-\delta)(c-1)(-\lambda_1\lambda_2)\rho(s) + \left(\frac{1+(1-\delta)^2}{\rho(s)}+(\lambda_1+\lambda_2)\right)\cdot\left(1-\frac{(c-1)\delta(2-\delta)}{c}\right)\\
        &= \delta(2-\delta)s(1-\delta)\sqrt{(c-1)(-\lambda_1\lambda_2)} +\\ &\left((1+(1-\delta)^2)\sqrt{(c-1)(-\lambda_1\lambda_2)}\frac{1}{s(1-\delta)}+(\lambda_1+\lambda_2)\right)\cdot\left(1-\frac{(c-1)\delta(2-\delta)}{c}\right)
    \end{align*}
    For any $\eps > 0$, we can choose $\delta$ small enough so that the above quantity is at least
    \[\lambda_1+\lambda_2+2\sqrt{(c-1)(-\lambda_1\lambda_2)}-\eps\] when $s = 1$ and at most
    \[\lambda_1+\lambda_2-2\sqrt{(c-1)(-\lambda_1\lambda_2)}+\eps\] when $s = -1$.

\end{proof}

\subsection{SDP solution for random additive lifts}

For $\eps>0$, consider $f_v^{(1)}$ constructed in the proof of \pref{lem:witness-vectors}, for which
\[
    \langle f_v^{(1)}, A_X f_v^{(1)}\rangle \ge (\lambda_1+\lambda_2+r_X-\eps)\|f_v^{(1)}\|^2/
\]
Let $L_\eps$ be an integer chosen such that the total $\ell_2$ mass of $\frac{f_v^{(1)}}{\|f_v^{(1)}\|}$ on vertices at distance greater than $L$ from $v$ is at most $\eps$.  Define $g_v$ as the vector obtained by zeroing out $\frac{f_v^{(1)}}{\|f_v^{(1)}\|}$ on vertices outside $B(v,L)$ and normalizing to make its norm 1, where $B(v,L)$ is the collection of vertices within distance $L$ of $v$.

For any $\eps'>0$, we can choose $\eps$ so that
\begin{align}
    \langle g_v, A_Xg_v\rangle \ge \lambda_1+\lambda_2+r_X-\eps' \label{eq:SDP-val}
\end{align}
$g_v$ enjoys the property of being determined by a constant number of vertices, $L_{\eps'}$.  For any instance graph $G$ such that there is a unique shortest nomadic walk between any pair of vertices $u$ and $v$, we can explicitly define
\[
    g_v(u) =
    \begin{cases}
        0 &\text{if $d(u,v)>L_{\eps'}$}\\
        C\prod\limits_{\{i,j\}\in\calP_{uv}}\frac{(1-\delta)(A_X)_{ij}}{\sqrt{(c-1)(-\lambda_1\lambda_2)}}&\text{$\calP_{uv}$ unique shortest nomadic walk from $u$ to $v$}
    \end{cases}
\]
where $C$ is a constant chosen so that $g_v$ has unit norm.

Recall that $\binstgraph_n$ is a random signed additive $n$-lift obtained from a sequence of atoms $\calA$.

\begin{definition}
    Let $G$ be a graph and let $\phi:E(G)\rightarrow\{\pm 1\}$ be a signing of the edges. We call a signing $\phi$ \textit{balanced} if for any cycle given by sequence of edges $e_1,\ldots,e_k$ in $E(H)$, we have $\phi(e_1)\cdots\phi(e_k)=1$.    
\end{definition}

We use $A_{\phi(G)}$ to denote the adjacency operator of $G$ signed with respect to $\phi$ --- i.e. $(A_{\phi(G)})_{uv}=\phi(\{u,v\})$
if $\{u,v\}$ is an edge and 0 otherwise.

\begin{lemma}\label{lem:achieve_balanced}
Suppose $\phi$ is a balanced signing of $G$. Then there exists a diagonal sign operator $D$ such that $A_{\phi(G)} = DA_GD^{\dagger}$.
\end{lemma}
\begin{proof}
Without loss of generality, assume $G$ is connected.  Take a spanning tree of $G$ and root it at some arbitrary vertex $r$.  Let $D_{rr} = 1$ and for $P_x$ a path from $r$ to $x$ let $D_{xx} = \prod_{e\in P_x}\phi(e)$.

It remains to verify that $DA_GD^{\dagger}=A_{\phi(G)}$.  Let $P$ be the path between $x$ and $y$ in the spanning tree.  By virtue of $\phi$ being balanced, we have $\phi(\{x,y\})\prod_{e\in P}\phi(e)=1$, which means $\phi(\{x,y\})=\prod_{e\in P}\phi(e)$.  Also, note that $\prod_{e\in P}\phi(e)$ is equal to $\prod_{e\in P_x}\phi(e)\prod_{e\in P_y}\phi(e)$, which is equal to $D_{xx}D_{yy}$.  Thus,
\[
    (A_{\phi(G)})_{ij} = \phi(\{i,j\})(A_G)_{ij} = D_{ii}D_{jj} (A_G)_{ij} = \left(DA_GD^{\dagger}\right)_{ij}
\]
which proves the claim.
\end{proof}

\begin{lemma}
    Let $X_D$ be the graph with the adjacency operator $DA_XD^{\dagger}$ where $D$ is a diagonal sign matrix.  There exists $D$ such that $X_D$ covers $\binstgraph_n$.
\end{lemma}
\begin{proof}
    When $\binstgraph_n$ is generated, (i) the sequence of atoms $\calA$ first undergoes an additive $n$-lift, and then, (ii) the atoms in the lifted graph are given a random balanced signing.  The intermediate graph $\widetilde{\binstgraph}_n$ between (i) and (ii) is covered by $X$ via a map $\pi:V(X)\to V(\widetilde{\binstgraph}_n)$.  Once (ii) is performed, construct $X'$ by taking $X$ and setting the signs on all edges in $\pi^{-1}(e)$ to the sign on $e$ for each $e\in E(\binstgraph_n)$.  $X'$ can be seen as a balanced signing applied on $X$, and hence there exists such a $D$ by \pref{lem:achieve_balanced}.
\end{proof}

\begin{definition}  \label{def:L-bad}
    Let $\pi$ be a covering map from appropriate $X_D$ to $\binstgraph_n$.  Call a vertex $v\in V(\binstgraph_n)$ \emph{$L$-bad} if $B(v, L)$ is not isomorphic to $B(v^*, L)$ where $v^*\in V(X_D)$ is such that $\pi(v^*) = v$.
\end{definition}

\begin{remark}  \label{rem:bad-equivalence}
    The condition of a vertex $v$ in $V(\binstgraph_n)$ being $L$-bad according to \pref{def:L-bad} is equivalent to the corresponding variable $v'$ in the constraint graph having a cycle in its distance $2L$-neighborhood.
\end{remark}

With the observation of \pref{rem:bad-equivalence} in hand, we can extract the following as a consequence of \cite{deshpande2018threshold}.
\begin{lemma}   \label{lem:few-bad}
    The number of $K$-bad vertices in graph $\binstgraph_n$ for constant $K$ is bounded by $O(\log n)$ with probability $1-o_n(1)$.
\end{lemma}

Construct a vector $\widetilde{g}_v$ for each vertex $v$ of $\binstgraph_n$.
\[
    \widetilde{g}_v =
    \begin{cases}
        e_v &\text{if $v$ is $L_{\eps'}$-bad}\\
        g_v &\text{otherwise}
    \end{cases}
\]

We are finally ready to prove \pref{thm:SDP-sol}.
\begin{proof}[Proof of \pref{thm:SDP-sol}]
    Let
    \[
        M_+ := \sum_{v\in V(\binstgraph_n)} \wt{g}_v\wt{g}_v^{\dagger}
    \]
    Writing out $(M_+)_{uu}$ for arbitrary $u$
    \begin{align*}
        (M_+)_{uu} &= \sum_{v\in V(\binstgraph_n)} \wt{g}_v(u)\wt{g}_v(u)\\
        &= \sum_{v\in V(\binstgraph_n)} \wt{g}_u(v)^2\\
        &= \|\wt{g}_u\|^2 = 1
    \end{align*}
    and writing out $\langle A_{\binstgraph_n}, M_+\rangle$ gives the following with probability $1-o_n(1)$.
    \begin{align*}
        \langle A_{\binstgraph_n}, M_+\rangle &= \sum_{v\in V(\binstgraph_n)} \langle \wt{g}_v, A_{\binstgraph_n} \wt{g}_v\rangle\\
        &= \sum_{\substack{v\in V(\binstgraph_n)\\ \text{$v$ is not $(L_{\eps}+1)$-bad}}} \langle \wt{g}_v, A_{\binstgraph_n} \wt{g}_v \rangle + \sum_{\substack{v\in V(\binstgraph_n)\\ \text{$v$ is $(L_{\eps}+1)$-bad}}} \langle \wt{g}_v, A_{\binstgraph_n} \wt{g}_v\rangle \\
        &\ge \sum_{\substack{v\in V(\binstgraph_n)\\ \text{$v$ is not $(L_{\eps}+1)$-bad}}} \lambda_1+\lambda_2+r_X-\eps' + \sum_{\substack{v\in V(\binstgraph_n)\\ \text{$v$ is $(L_{\eps}+1)$-bad}}} c(\lambda_1\lambda_2) &\text{(by \pref{eq:SDP-val})}\\
        &\ge (n-O(\log n))(\lambda_1+\lambda_2+r_X-\eps') - O(\log n) &\text{(by \pref{lem:few-bad})}\\
        &= (1-o_n(1))(\lambda_1+\lambda_2+r_X-\eps')n
    \end{align*}
    The desired inequality on $\langle A_{\binstgraph_n}, M_+\rangle$ can be obtained by choosing $\eps'$ small enough and $n$ large enough.  The inequality on $\langle A_{\binstgraph_n}, M_-\rangle$ can be proved by repeating the whole section and proof by constructing vectors $\wt{g}_v$ from $f_v^{(-1)}$.
\end{proof}

%% file: bordenave2.tex
\section{Friedman/Bordenave for additive lifts} \label{sec:bordenave}
%\rnote{throughout gotta give more credit to Bordenave, Dumitriu, and DMOSS at various places}

\begin{theorem}                                     \label{thm:bordenave1}
    Let $\atomlist = (A_1, \dots, A_c)$ be a sequence of $r$-vertex atoms with edges weights~$\pm 1$.  Let $|\instgraph_1|$  denote the instance graph $\atomlist(K_{r,c})$ associated to the base constraint graph when the edge-signs are deleted (i.e., converted to~$+1$), and let $|B_1|$ denote the associated nomadic walk matrix.  Also, let $\bcongraph_n$ denote a random $n$-lifted constraint graph and $\binstgraph_n = \atomlist(\bcongraph_n)$ an associated instance graph with $1$-wise uniform negations~$(\bxi^f_{ii'})$.  Finally, let $\bB_n$ denote the nomadic walk matrix for~$\binstgraph_n$.  Then for every constant $\eps > 0$,
    \[
        \Pr[\rho(\bB_n) \geq \sqrt{\rho(|B_1|)} + \eps] \leq \delta,
    \]
    where $\delta = \delta(n)$ is $o_{n \to \infty}(1)$. % \rnote{I think $\delta = O(1/n^{.99})$ and $\eps = \wt{O}(1/\log n)$.}
\end{theorem}
\begin{remark}
    It might seem that our bound involving $|B_1|$ may be poor, given that it ignores sign information from the atoms.  However, it is in fact sharp, and the reason is that the main contribution to $\rho(\bB_n)$ when using the Trace Method is from walks in which almost all edges are traversed twice.  And if an edge is traversed twice, it of course does not matter if its sign is~$-1$ or~$+1$. %\rnote{I'm not really sure that this ``story'' is actually getting at the truth of the matter.}
\end{remark}
\begin{remark}
    In fact, it is evident from the theorem statement that without loss of generality we may assume that the atoms are unweighted --- i.e., that all weights are~$+1$.  The reason is that for each constraint~$f$ in group~$j$, if we multiply $\bxi^f_{ii'}$ by the fixed value $A_j[i,i']$, the resulting signs remain $1$-wise uniform --- and this has the effect of eliminating all signs from the atoms.  Thus henceforth we will indeed assume that the original atoms are all unweighted.
\end{remark}

% \rnote{this explanation is maybe overly brief, but I wanted to write it}
\paragraph{The idea of Friedman/Bordenave proofs.} The standard method for trying to prove a theorem such as \pref{thm:bordenave1} involves applying the Trace Method to~$\bB_n$.  Since $\bB_n$ is not a self-adjoint operator, a natural way to do this is to consider $\tr(\bB_n^\ell {\bB_n^*}^\ell)$ for some large~$\ell$.  Roughly speaking, this counts the number of closed walks that walk nomadically in~$\binstgraph_n$ for the first~$\ell$ steps, and then walk nomadically in the \emph{reverse} of~$\binstgraph_n$ for the next~$\ell$ steps.  A major difficulty is the following: the Trace Method naturally incurs an ``extra'' factor of~$n$, and to overcome this one wants to choose $\ell \gg \log n$.  However, $\Theta(\log n)$ is precisely the radius at which random constraint graphs become dramatically non-tree-like; i.e., they are likely to encounter nontrivial cycles.  Based on Friedman's work, Bordenave overcomes this difficulty as follows:  First, $\ell$ is set to $c \log n$ for some small positive constant $c > 0$.  Nomadic walks of this length may well encounter cycles, but one can show that with high probability, they will not encounter \emph{tangles} --- meaning, \emph{more than one} cycle in a radius of~$\ell$.  (This crucial concept of ``tangles'' was isolated by Friedman and refined by Bordenave.)  Now we set $k = \omega_n(1)$ to be a slowly growing quantity and consider length-$2k \ell$ walks formed by doing $\ell$ nomadic steps, then $\ell$ nomadic reverse-steps, all $k$ times in succession.  In other words, we consider $\tr((\bB_n^\ell {\bB_n^*}^\ell)^k)$.  On one hand, since $2k\ell \gg \log n$, bounding this quantity will be sufficient to overcome the $n$-factor inherent in the Trace Method.  On the other hand, using tangle-freeness at radius~$\ell$ along with very careful combinatorial counting allows us to bound the number of closed length-$2k\ell$ walks.

Our proof follows this methodology and draws ideas from Bordenave's original proof from \cite{bordenave2015new} as well as \cite{deshpande2018threshold} and \cite{brito2018spectral}. However, our main technical lemma, \pref{lem:fb}, uses a new tool that takes advantage of the random negations our model employs that simplifies the equivalent proofs in the three mentioned papers and also allows us to generalize it to our model.

\subsection{Trace Method setup, and getting rid of tangles}
To begin carrying out this proof strategy, we first define tangle-freeness.
\begin{definition}[Tangles-free]
    Let $G$ be an undirected graph.  A vertex $v$ is said to be \emph{$\ell$-tangle-free within $G$} if the subgraph of~$G$ induced by~$v$'s distance-$4\ell$ neighborhood contains at most one cycle.\footnote{We chose the factor $4$ here for ``safety''. For quantitative aspects of our theorem, constant factors on~$\ell$ will be essentially costless.}
    %Finally, if $\instgraph$ is an instance graph based on constraint graph $\congraph$, we say that $\instgraph$ is \emph{$\ell$-atom-tangle-free} whenever $\congraph$ is $\ell$-tangle-free.\rnote{do we end up using this notion?}
\end{definition}
It is straightforward to show that random lifts have all vertices $\Theta(\log n)$-tangle-free; we can quote the relevant result directly from Bordenave~(Lemma 27 from \cite{bordenave2015new}):
\begin{proposition}                                     \label{prop:tanglefree}
    There is a universal constant $\kappa > 0$ depending only on $r$, $c$ such that, for $\ell  = \kappa \log n$, a random $n$-lift $\bcongraph$ of $K_{r,c}$ has all vertices $\ell$-tangle free, except with probability~$O(1/n^{.99})$.
\end{proposition}
%In our main \pref{thm:bordenave1}, if $\bcongraph_n$ is not $\ell$-tangle-free then we will simply give up, losing negligible probability. Otherwise, $\binstgraph_n$ is $\ell$-atom-tangle-free.  Under this assumption, we would like to get an expression $\tr((\bB_n^\ell {\bB_n^*}^\ell)^k)$.

We now begin the application of the Trace Method. We have:
\begin{gather}
\begin{aligned}
    \tr((\bB_n^\ell {\bB_n^*}^\ell)^k)
    &= \sum_{\vec{e}_0, \dots, \vec{e}_{2k\ell-1}, \vec{e}_{2k\ell} = \vec{e}_0} \bB_n[\vec{e}_0, \vec{e}_1] \cdots %\bB_n[\vec{e}_1, \vec{e}_2]
    \bB_n[\vec{e}_{\ell-1}, \vec{e}_{\ell}] \bB_n^*[\vec{e}_{\ell}, \vec{e}_{\ell+1}] \cdots \bB_n^*[\vec{e}_{2\ell-1}, \vec{e}_{2\ell}] \cdots \bB_n^*[\vec{e}_{2k\ell-1}, \vec{e}_{2k\ell}] \\
    &= \sum_{\vec{e}_0, \dots, \vec{e}_{2k\ell-1}, \vec{e}_{2k\ell} = \vec{e}_0} \bB_n[\vec{e}_0, \vec{e}_1] \cdots
    \bB_n[\vec{e}_{\ell-1}, \vec{e}_{\ell}] \bB_n[\vec{e}_{\ell+1}, \vec{e}_{\ell}] \cdots \bB_n[\vec{e}_{2\ell}, \vec{e}_{2\ell-1}] \cdots \bB_n[\vec{e}_{2k\ell}, \vec{e}_{2k\ell-1}]
\end{aligned} \nonumber\\
= \sum \Wt(e_1) N_{\vec{e}_0, \vec{e}_1} \cdots
    \Wt(e_\ell) N_{\vec{e}_{\ell-1}, \vec{e}_{\ell}} \Wt(e_\ell) N_{\vec{e}_{\ell}^{-1}, \vec{e}_{\ell+1}^{-1}} \cdots \Wt(e_{2\ell-1}) N_{\vec{e}_{2\ell-1}^{-1}, \vec{e}_{2\ell}^{-1}} \cdots \Wt(e_{2k\ell-1})N_{\vec{e}_{2k\ell-1}^{-1}, \vec{e}_{2k\ell}^{-1}}, \label{eqn:trace0}
\end{gather}
where $\Wt(e)$ is the sign on edge~$e$ coming from the random $1$-wise negations (it is the same for both directed versions of the edge), and where $N_{\vec{e}, \vec{f}}$ is an indicator that $(\vec{e}, \vec{f})$ forms a length-$2$ nomadic walk.  Roughly speaking, this quantity counts (with some $\pm 1$ sign) closed walks in $\binstgraph_n$ consisting of $2k$ consecutive nomadic walks of length~$\ell$.  However, there is some funny business concerning the joints between these nomadic walks.  To be more precise, in each of the $2k$ segments we have a nomadic walk of $\ell+1$ edges; and, the last edge in each segment must be the reverse of the first edge in the subsequent segment. We will call these necessarily-duplicated edges ``spurs''. Furthermore, when computing the sign with which the closed walk is counted, spurs' signs are counted either zero times or twice, depending on the parity of the segment.  Hence they are effectively discounted, since $(-1)^2 = (-1)^0 = +1$.  Let us make some definitions encapsulating all of this.
\begin{definition}[Nomadic linkages, and spurs] \label{def:nomlink}
    In an instance graph, a \emph{$(2k \times \ell)$-nomadic linkage}~$\calL$ is the concatenation of $2k$ many nomadic walks (``segments''), each of length~$\ell+1$, in which the last directed edge of each walk is the reverse of first directed edge of the subsequent walk (including wrapping around from the $2k$th segment to the $1$st). These $2k$ directed edges which are necessarily the reverse of the preceding directed edge are termed \emph{spurs}.  The \emph{weight} of $\calL$, denoted $\Wt(\calL)$, is the product of the signs of the non-spur edges in~$\calL$.
\end{definition}
\begin{definition}[Nonbacktracking $\calA$-linkages]
    Recall that, strictly speaking, the nomadic property requires ``remembering'' which atom each edge comes from.  Thus the $\calL$ above is really associated to what we will call a \emph{$(2k \times 2\ell)$-nonbacktracking $\atomlist$-linkage} --- call it $\calC$ --- in the underlying constraint graph.  Formally:
    \begin{itemize}
        \item (``linkage'') $\calC$ is a closed concatenation of $2k$ walks (called ``segments'') in the constraint graph, each consisting of $\ell+1$  length-$2$ variable-constraint-variable subpaths.   The last such length-$2$ subpath in each segment (``spur'') is equal to (the reverse of) the first length-$2$ subpath in the subsequent segment (including wraparound from the $2k$th segment to the $1$st).
        \item (``$\atomlist$-linkage'') For each length-$2$ subpath $(v, f, v')$ in~$\calC$, where $v$ is in variable group~$i$, $f$ is in constraint group~$j$, and $v'$ is in variable group~$i'$, it holds that $\{i,i'\}$ is an edge in~$A_j$.
        \item (``nonbacktracking'') Each of the $2k$ segments is a nonbacktracking walk of length $2(\ell+1)$ in the constraint graph.
    \end{itemize}
    We write $\Wt(\calC) \in \{\pm 1\}$ for the weight of the associated nomadic linkage in the instance graph.
\end{definition}
\noindent Given these definitions, \pref{eqn:trace0} tells us:
\begin{equation}    \label{eqn:traceee}
    \tr((\bB_n^\ell {\bB_n^*}^\ell)^k) = \sum_{\substack{(2k \times 2\ell)\textnormal{-nonbacktracking} \\ \atomlist\textnormal{-}\textnormal{linkages }\calC \textnormal{ in } \bcongraph_n}} \Wt(\calC).
\end{equation}
Next, we make the observation that \emph{if} $\bcongraph_n$ proves to have all vertices $\ell$-tangle-free, then we would get the same result if we only summed over ``externally tangle-free'' linkages.
\begin{definition}[Externally tangle-free linkages]  \label{rem:2k2l-tangle}
    We say that a $(2k \times 2\ell)$-nonbacktracking linkage in a constraint graph $\congraph_n$ is \emph{externally $\ell$-tangle-free} if every vertex it touches is $\ell$-tangle-free within~$\congraph_n$.  (The ``externally'' adjective emphasizes that we are concerned with cycles not just within the linkage's edges, but also among nearby edges of~$\congraph_n$.)
\end{definition}
Thus in light of \pref{prop:tanglefree} we have:
\begin{lemma} \label{lem:trace3}
    Provided $\ell \leq \kappa \log n$ for a certain universal $\kappa > 0$, we get that $\tr((\bB_n^\ell {\bB_n^*}^\ell)^k)  = \bS$ holds except with probability $O(1/n^{.99})$ , where
    \[
        \bS \coloneqq  \sum_{\substack{(2k \times 2\ell)\textnormal{-nonbacktracking} \\ \textnormal{externally $\ell$-tangle-free} \\ \atomlist\textnormal{-}\textnormal{linkages } \calC \textnormal{ in } \bcongraph_n}} \Wt(\calC).
    \]
\end{lemma}
In order to apply Markov's inequality later, we will need the following technical claim:
\begin{claim}                                       \label{claim:S-nonneg}
    $\bS$ is a nonnegative random variable.
\end{claim}
\begin{proof}
    Given $\binstgraph_n$, recall that
    \[
        \bB_n^\ell[\vec{e}, \vec{f}] = \sum_{\substack{\text{nomadic walks} \\ \vec{e} = \vec{e}_0, \vec{e}_1, \dots, \vec{e}_\ell = \vec{f} \text{ in $\binstgraph_n$}}} \Wt(e_1) \Wt(e_2) \cdots \Wt(e_\ell).
    \]
    Using a key idea of Bordenave (based on the ``selective trace'' of Friedman), define the related operator $\bB_n^{(\ell)}$ via
    \[
        \bB_n^{(\ell)}[\vec{e}, \vec{f}] = \sum_{\substack{\text{\emph{externally $\ell$-tangle-free} nomadic walks} \\ \vec{e} = \vec{e}_0, \vec{e}_1, \dots, \vec{e}_\ell = \vec{f} \text{ in $\binstgraph_n$}}} \Wt(e_1) \Wt(e_2) \cdots \Wt(e_\ell),
    \]
    where again the walk is said to be ``externally $\ell$-tangle-free'' if every vertex it touches is $\ell$-tangle-free with $\bcongraph_n$.  Then very similar to the analysis that gave us \pref{eqn:trace0} and \pref{eqn:traceee}, we get that
    \[
        \bS = \tr((\bB_n^{(\ell)} (\bB_n^{(\ell)})^*)^k).
    \]
    Thus $\bS$ is visibly always nonnegative, being the trace of the $k$th power of the positive semidefinite matrix $\bB_n^{(\ell)} (\bB_n^{(\ell)})^*$.
\end{proof}
With these results in place, we can proceed to the main goal of the Trace Method: bounding $\E[\bS]$.  Such a bound can be used in the following lemma:
\begin{lemma}                                       \label{lem:finish-trace}
    Assume that $\ell \leq \kappa \log n$ and $k \ell = \omega(\log n)$.  Then from $\E[\bS] \leq R$ we may conclude that
    $
        \specrad(\bB_n) \leq (1+o_n(1)) \cdot R^{\frac{1}{2k\ell}}
    $ holds, except with probability $O(1/n^{.99})$.
\end{lemma}
\begin{proof}
    Let $\bT = \tr((\bB_n^\ell {\bB_n^*}^\ell)^k)$.  On one hand, with $\lambda$ denoting eigenvalues and $\sigma$ denoting singular values, we have
    \[
        \bT \geq \lambda_{\text{max}}((\bB_n^\ell {\bB_n^*}^\ell)^k) = \lambda_{\text{max}}\parens*{\sqrt{\bB_n^\ell {\bB_n^*}^\ell}}^{2k} = \sigma_{\text{max}}(\bB_n^{\ell})^{2k} \geq \specrad(\bB_n^\ell)^{2k} = \specrad(\bB_n)^{2k\ell}.
    \]
    On the other hand, since $\bS$ is a nonnegative random variable (\pref{claim:S-nonneg}), we can apply Markov's Inequality to deduce that $\bS \leq n\cdot R$ except with probability at most~$1/n$.  Now from \pref{lem:trace3} we may infer that except with probability $O(1/n^{.99})$,
    \[
        \bT = \bS \leq n \cdot R \quad\implies\quad \specrad(\bB_n)^{2k\ell} \leq n \cdot R.
    \]
     The result now follows by taking $2k\ell$-th roots.
\end{proof}

\subsection{Eliminating singletons, and reduction to counting}
Our next step toward bounding $\E[\bS]$ is typical of the Trace Method:  Rather than first choosing~$\bcongraph_n$ randomly and then summing over the linkages therein, we instead sum over all \emph{potentially-appearing} linkages and insert an indicator that they actually appear in the realized random constraint graph.  Defining
\[
    \calK_n = \text{the ``complete'' constraint graph with $cn$ constraint vertices and $rn$ variable vertices},
\]
this means that
\begin{equation}                                        \label{eqn:trace2}
    \bS =  \sum_{\substack{(2k \times 2\ell)\text{-nonbacktracking} \\ \atomlist\textnormal{-}\text{linkages } \calC \text{ in } \calK_n}} 1[\calC \text{ is in } \bcongraph_n] \cdot 1[\calC \text{ is externally $\ell$-tangle-free within } \bcongraph_n] \cdot \Wt_{\binstgraph_n}(\calC).
\end{equation}
Here we wrote $\Wt_{\binstgraph_n}(\calC)$ to emphasize that even once $\calC$ is in $\bcongraph_n$ and is externally $\ell$-tangle-free, its weight is still a random variable arising from the $1$-wise uniform negations.  These negations will create another simplification (one not available to Friedman/Bordenave).  For this we will need another definition:
\begin{definition}[Singleton-free $\calC$'s]
    Let  $\calC$ be a  $(2k \times 2\ell)$-nonbacktracking circuit in $\calK_n$.  If there is an atom vertex that is passed through exactly once, we call it a \emph{singleton}.  If $\calC$ contains no singleton, we call it \emph{singleton-free}.
\end{definition}
Referring to \pref{eqn:trace2}, consider $\E[\bS]$.  If $\calC$ contains any singleton, then it will contribute~$0$ to this expectation.  The reason is that, provided $\calC$ appears in $\bcongraph_n$ and is externally $\ell$-tangle-free therein, the $1$-wise uniform negations will assign a uniformly random $\pm 1$ sign to the edge engendered by $\calC$'s singleton, and this sign will be independent of all other signs that go into $\Wt_{\binstgraph_n}(\calC)$.  On the other hand, when $\calC$ is singleton-free, we will simply upper-bound the (conditional) expectation of $\Wt_{\binstgraph_n}(\calC)$ by~$+1$.  We conclude that
\begin{equation}    \label{eqn:restrict}
    \E[\bS] \leq\sum_{\substack{(2k \times 2\ell)\text{-nonbacktracking} \\ \textit{singleton-free} \\ \atomlist\textnormal{-}\text{linkages } \calC \text{ in } \calK_n}} \Pr[\calC \text{ is in } \bcongraph_n \text{ and is externally $\ell$-tangle-free therein}].
\end{equation}
Let us now begin to simplify the probability calculation.
\begin{definition}[$E(\calC)$, $V(\calC)$, $G(\calC)$]
    Let $\calC$ be a $(2k \times 2\ell)$-nonbacktracking $\atomlist$-linkage in $\calK_n$. Write $E(\calC)$ for the set of undirected edges in $\calK_n$ formed by ``undirecting'' all the directed edges in~$\calC$ (this includes reducing from a multiset to a set, if necessary). Then let $G(\calC)$ denote the undirected subgraph of $\calK_n$ induced by $E(\calC)$, and write $V(\calC)$ for its vertices.
\end{definition}
Let's simplify the ``tangle-freeness'' situation.
\begin{definition}[Internal tangle-free linkages]
    We say that a $(2k \times 2\ell)$-nonbacktracking linkage $\calC$ in $\calK_n$ is \emph{internally $\ell$-tangle-free} if every vertex it touches is $\ell$-tangle-free \emph{within $G(\calC)$}.
\end{definition}
We certainly have:
\begin{multline*}
    \text{linkage $\calC$ not even internally $\ell$-tangle-free} \\ \implies\quad \Pr[\calC \text{ is in } \bcongraph_n \text{ and is externally $\ell$-tangle-free therein}] = 0.
\end{multline*}
Thus we can restrict the sum in \pref{eqn:restrict} to internally $\ell$-tangle-free linkages.  Having done that, we will upper bound the sum by dropping this insistence on \emph{external} tangle-freeness.  Thus
\begin{equation}    \label{eqn:restrict2}
    \E[\bS] \leq\sum_{\substack{(2k \times 2\ell)\text{-nonbacktracking} \\ \text{\emph{interally $\ell$-tangle-free}, singleton-free} \\ \atomlist\textnormal{-}\text{linkages } \calC \text{ in } \calK_n}} \Pr[\calC \text{ is in } \bcongraph_n].
\end{equation}

We will now bound $\Pr[\calC \text{ is in } \bcongraph_n]$, so as to reduce all our remaining problems to counting.  Towards this, recall that $\bcongraph_n$ is a random $n$-lift of the complete graph $K_{r,c}$.  One thing this implies is that every group-$i$ variable-vertex in $\bcongraph_n$ will have exactly one edge to each of $c$~groups of constraint-vertices, and vice versa.  Let us codify the $\calC$'s that don't flagrantly violate this property:
\begin{definition}[Valid $\calC$'s]
    We say a $(2k \times 2\ell)$-nonbacktracking $\atomlist$-linkage $\calC$ in $\calK_n$ is \emph{valid} if $G(\calC)$ has the property that every variable-vertex in it is connected to at most~$1$ constraint-vertex from each of the $c$~groups, and each constraint-vertex is connected to at most~$1$ variable-vertex from each of the $r$~groups.
\end{definition}
Evidently,  $\Pr[\calC \text{ is in } \bcongraph_n] = 0$ if $\calC$ is invalid.  Thus from \pref{eqn:restrict2} we can deduce:
\begin{equation} \label{eqn:restrict3}
    \E[\bS] \leq \sum_{\substack{(2k \times 2\ell)\text{-nonbacktracking} \\ \text{\emph{valid}, internally $\ell$-tangle-free, singleton-free} \\ \atomlist\textnormal{-}\text{linkages } \calC \text{ in } \calK_n}} \Pr[\calC \text{ is in } \bcongraph_n].
\end{equation}
Next, it is straightforward to show the following lemma (see Proposition A.8 of \cite{deshpande2018threshold} for essentially the same observation):
\begin{lemma}                                       \label{lem:validprob}
    If $\calC$ is a valid $(2k \times 2\ell)$-nonbacktracking $\atomlist$-linkage in $\calK_n$, and $k\ell = o(\sqrt{n})$, then \[\Pr[\calC \textnormal{ is in } \bcongraph_n] = (1+o_n(1)) \cdot n^{-|E(\calC)|}.\]
\end{lemma}
\begin{proof}
    (Sketch.) Proceed through the edges in $E(\calC)$ in an arbitrary order.  Each has approximately a $1/n$ chance of appearing in $\bcongraph_n$, even conditioned on the appearance of the preceding edges.  For example, this is exactly true for the first edge.  For subsequent edges $e = \{u,v\}$, validity ensures that no preceding edge already connects~$u$ to a vertex in $v$'s part, or vice versa.  Thus the conditional probability of~$e$  appearing in $\bcongraph_n$ is essentially the probability that a particular edge appears in a random matching on $n+n$ vertices (which is $1/n$), except that a ``small'' number of vertex pairs may already have been matched.  This ``small'' quantity is at most $|E(\calC)| \leq 4k\ell$, so the $1/n$ probability becomes $1/(n-4k\ell)$ at worst.  Multiplying these conditional probabilities across all $|E(\calC)|$ edges yields a quantity that is off from $n^{-|E(\calC)|}$ by a factor of at most $(1+O(k\ell)/n)^{4k\ell} \leq 1+o_n(1)$, the inequality using $(k\ell)^2 = o(n)$.
\end{proof}

Combining this lemma with \pref{eqn:restrict3} and \pref{lem:finish-trace}, we are able to reduce bounding $\specrad(\bB_n)$ to a counting problem:
\begin{lemma}                                       \label{lem:finish-trace2}
    Assume that $\ell \leq \kappa \log n$ and $\omega(\log n) < k\ell < o(\sqrt{n})$.  Then except with
    probability~$O(1/n^{.99})$,
    \[
        \specrad(\bB_n) \leq (1+o_n(1)) \cdot R^{\frac{1}{2k\ell}}, \quad \text{where } R \coloneqq \sum_{\substack{(2k \times 2\ell)\textnormal{-nonbacktracking} \\ \textnormal{valid, internally $\ell$-tangle-free, singleton-free} \\ \atomlist\textnormal{-linkages } \calC \textnormal{ in } \calK_n}} n^{-|E(\calC)|}.
    \]
\end{lemma}

\subsection{Tangle-free, singleton-free linkages are nearly duplicative}
Our goal in this subsection is to show that each linkage $\calC$ we sum over in \pref{lem:finish-trace2} is ``nearly duplicative'': the number of variable-vertices is at most $(1+o(1))k\ell$, and the same is true of constraint-vertices --- even though the obvious a priori upper bound for each of them is~$2k\ell$.  This factor-$\frac12$ savings is precisely the source of the square-root in \pref{thm:bordenave1}.  We begin with a graph-theoretic lemma and then deduce the nearly-duplicative property.
\begin{lemma}                                       \label{lem:deg3}
    Let $\calC$ be a $(2k \times 2\ell)$-nonbacktracking, internally $\ell$-tangle-free linkage in $\calK_n$. Assume $\log(k\ell) = o(\ell)$.  Then $G(\calC)$ has at most $O(k \log(k\ell))$ vertices of degree exceeding~$2$.
\end{lemma}
\begin{proof}
    For brevity, let us write $G = G(\calC)$, $w = |V(\calC)|$, and note that we have a trivial upper bound of $w \leq 4k\ell$.  Let $t$ denote the number of  cycles of length at most~$\ell$ in~$G$.  By deleting at most~$t$ edges, we can form a graph~$\wt{G}$ with girth at least~$\ell$. A theorem of Alon, Hoory, and Linial \cite{alon2002moore} implies that any (possibly irregular) graph with $w$ vertices and girth at least~$\ell$ must have average degree at most $2 + O(\log(w)/\ell)$ (this uses $\log(w) = o(\ell)$).  Thus $\wt{G}$ has such a bound on its average degree.  After restoring the deleted edges, we can still conclude that the average degree in~$G$ is at most $2 + O(\log(w)/\ell) + \frac{2t}{w}$.  Writing $w_1, w_2, w_{3^+}$ for the number of vertices in~$G$ of degree~$1$, $2$, and $3$-or-more respectively, this means
    \begin{gather*}
         2 + O(\log(w)/\ell) + \frac{2t}{w} \geq \frac{w_1 + 2w_2 + 3w_{3^+}}{w} = \frac{w_1 + 2(w - w_1 - w_{3^+}) + 3w_{3^+}}{w} = 2 - \frac{w_1}{w} + \frac{w_{3^+}}{w} \\
         \implies \quad w_{3^+} \leq O(w \log(w)/\ell) + w_1 + 2t.
    \end{gather*}
    The first term here is $O(k \log(k\ell))$ as desired, since $w \leq 4k\ell$.  We will also show the next two terms are~$O(k)$.  Regarding $w_1$, degree-$1$ vertices in $G$ can only arise from the spurs of~$\calC$, and hence $w_1 \leq 2k$.  Finally, $2t \leq O(k)$ follows from the below claim combined with $w \leq 4k\ell$:
    \begin{equation}    \label{eqn:t}
        t \leq \frac{w}{2\ell} + 1.
    \end{equation}
    We establish \pref{eqn:t} using the tangle-free property of~$\calC$.  Recall that $t$ is the number of ``short'' cycles in~$G$, meaning cycles of length at most~$\ell$.  By the $\ell$-tangle-free property of~$\calC$ (recalling the factor~$4$ in its definition), every $v \in V$ has at most one short cycle within distance $3\ell$ of it.  Thus if we choose paths in~$G$ that connect all short cycles (recall $G$ is connected), then to each short cycle we can uniquely charge at least $3\ell-1 \geq 2\ell$ vertices from these paths.  It follows that $w = |V| \geq 2\ell(t-1)$, establishing \pref{eqn:t}.
\end{proof}
\begin{corollary}                                       \label{cor:dupe}
    In the setting of \pref{lem:deg3}, assume also that $\calC$ is singleton-free and valid.  Then the number of variable-vertices $\calC$ visits is at most $k\ell + O(k \log(k\ell))$, and the same is true of constraint-vertices.
\end{corollary}
\begin{proof}
    Think of $\calC$ as a succession of $2k(\ell + 1)$ ``two-steps'', where a two-step is a length-$2$ directed path going from a variable-vertex, to a constraint-vertex, to a (distinct) variable-vertex.  Call two such two-steps ``duplicates'' if they use the same three variables (possibly going in the opposite direction).  We claim that ``almost all'' two-steps have at least one duplicate.  To see this, consider the constraint-vertex in some two-step~$a$.  Since $\calC$ is singleton-free, at least one other two-step~$b$ must pass through the constraint-vertex of~$a$.  If $b$~is not a duplicate of~$a$, then this constraint-vertex will have degree exceeding~$2$ in $G(\calC)$.  By \pref{lem:deg3} there are at most $O(k \log(k\ell))$ such constraint-vertices.  Further, by validity each constraint-vertex can support at most $\binom{r}{2} = O(1)$ unduplicated two-steps.  Thus at most $O(k \log(k\ell))$ of the $2k(\ell + 1)$ two-steps are unduplicated.

    Now imagine we walk through the two-steps of~$\calC$ in succession.  Each two-step can visit at most one ``new'' variable-vertex and one ``new'' constraint-vertex.  However each two-step which is a duplicate of a previously-performed two-step visits no new vertices.  Among the $2k(\ell + 1)$ two-steps, at most $O(k \log(k\ell))$ are unduplicated.  Thus at least $(2k(\ell + 1) - O(k \log(k\ell)))/2 = k(\ell + 1) - O(k \log(k\ell))$ two-steps are duplicates of previously-performed two-steps.  It follows that at most $k(\ell + 1)  + O(k \log(k\ell))$ two-steps visit any new vertex.  This completes the proof.
\end{proof}

\subsection{The final countdown}
We now wish to count the objects summed in the definition of~$R$ from \pref{lem:finish-trace2}. The remainder of this section will be devoted to proving:
\begin{theorem}     \label{thm:C}
    For every $\eps > 0$, except with probability~$O(1/n^{.99})$,
    
    \[
        \rho(\bB_n)\le (1+o_n(1))\cdot (1+\eps) \cdot\sqrt{\rho(|B_1|)}.
    \]
\end{theorem}

\noindent The bulk of the technical matter in the proof of \pref{thm:C} will involve analyzing
    \begin{equation}    \label{eqn:C}
        (2k \times 2\ell)\textnormal{-nonbacktracking, valid, internally $\ell$-tangle-free, singleton-free, } \atomlist\textnormal{-linkages } \calC
    \end{equation}
in $\calK_n$.

\begin{definition}[Steps: stale, fresh, and boundary]   \label{def:steptype}
    We call each of the $4k(\ell + 1)$ directed edges from which~$\calC$ is composed a \emph{step}.  If we imagine traversing these steps in order, they ``reveal'' vertices and edges of~$G(\calC)$ as we go along.  We call a step \emph{stale} if the edge it traverses was previously traversed in~$\calC$ (in some direction).  Note that both endpoints of the edge must also have been previously visited.  Otherwise, if the step traverses a ``new'' edge, it will be designated either ``fresh'' or ``boundary''.  It is designated \emph{fresh} if the vertex it reaches was never previously visited in~$\calC$.  Otherwise, the step is \emph{boundary}; i.e., the step goes between two previously-visited vertices, but along a new edge.  For the purposes of defining fresh/boundary, we specify that the initial vertex of~$\calC$ is always considered to be ``previously visited''.
\end{definition}
The following facts are immediate:
\begin{fact}                                        \label{fact:f-count}
    The number of fresh steps in~$\calC$ is $|V(\calC)|-1$.  (The $-1$ accounts for the fact that the initial vertex is considered ``previously visited''.)  Since the number of fresh and boundary steps together is $|E(\calC)|$, it follows that the number of boundary steps is $|E(\calC)| - |V(\calC)| + 1$.
\end{fact}
\newcommand{\lkgs}{\mathrm{Lkgs}} \newcommand{\Lkgs}{\mathrm{Lkgs}}
\begin{definition}
    We write $\Lkgs(f,b)$ for the collection of linkages as in \pref{eqn:C} having exactly $f$~fresh edges and $b$~boundary edges.
\end{definition}
Our goal is to show:
\begin{lemma}                                   \label{lem:fb}
    For every $\hat{\rho} > \rho(|B_1|)$ we have:
    
    $$|\Lkgs(f,b)| \leq \poly(k,\ell)^{b+k} \cdot n^{f+1} \cdot \hat{\rho}^{f/2}$$
    
    \noindent where the constants in the $\poly$ factor depend on $\hat{\rho}$.
\end{lemma}

Before proving this lemma, observe that many linkages are the same modulo the labels between $1$ and $n$ that are defined by the lifting. To make this formal we first introduce some notation and follow by using it to aid in the proof of \pref{lem:fb}.  

Given a linkage $\calC$ we write $\calC = ((v_1, i_1), (v_2, i_2), \ldots, (v_{4k(\ell + 1)}, i_{4k(\ell + 1)}))$, where $(v_j, i_j)$ are vertices from~$\calK_n$ and $v_j$ indicates the base vertex (from $K_{r, c}$) and $i_j$ is an integer (between~$1$ and $n$) that indicates the lifted copy. This notation means that $\calC$ traverses this sequence of vertices in this order.

\begin{definition}[Isomorphism of linkages]  \label{def:unnum}
    Given two linkages~$\calC$ and $\calC'$ that visit $|V(\calC)| = |V(\calC')|$ vertices, we say they are \emph{isomorphic} if are the same modulo the labels between $1$ and $n$ that are defined by the lifting. Formally, letting $\calC = ((v_1, i_1), \ldots, (v_{4k(\ell + 1)}, i_{4k(\ell + 1)}))$ and $\calC' = ((v_1', i_1'), \ldots, (v_{4k(\ell + 1)}', i_{4k(\ell + 1)}'))$, there exist permutations $\pi_v$ on $[n]$ for each $v \in V(K_{r, c})$ 
    such that for all $j$ we have $v'_j = v_j$ and $i'_j = \pi_{v_j}(i_j)$.
\end{definition}

This isomorphism relation induces equivalence classes for which we want to assign representative elements. We do so as follows.

\begin{definition}[Canonical linkages]  \label{def:canlink}
    A linkage~$\calC$ is said to be \emph{canonical} if for every vertex $v \in K_{r, c}$, if $\calC$ visits $j$ distinct lifted copies of $v$ then it first visits $(v, 1)$, then $(v, 2), \ldots$, and finally $(v, j)$.  We write $\Lkgs^c(f,b)$ for the collection of \emph{canonical} linkages as in \pref{eqn:C} having exactly $f$~fresh steps and $b$~boundary steps.
\end{definition}

%\begin{proposition}                                     \label{prop:canlkgs-bound}
%    The number of linkages having exactly $f$~fresh steps and $b$~boundary steps $|\Lkgs(f,b)|$ is at most $n^{f+1}|\Lkgs^c(f,b)|$.
%\end{proposition}
%\begin{proof}
%    Let's bound the number of $\calC'$ that are isomorphic to some linkage $\calC$. From the definition of isomorphic, we have a sequence of permutations on $n$ elements for which we are fixing some $|V(\calC)|$ elements. This is clearly upper bounded by $n^{|V(\calC)|}$, which is $n^{f + 1}$ by \pref{fact:f-count}. This implies the result.
%\end{proof}
\begin{proposition}                                     \label{prop:canlkgs-bound}
    $\displaystyle |\Lkgs(f,b)| \leq n^{f+1}|\Lkgs^c(f,b)|.$
\end{proposition}
\begin{proof}
    It suffices to show that for every canonical linkage $\calC \in \Lkgs^c(f,b)$, it has at most $n^{f+1}$ isomomorphic linkages $\calC' \in \Lkgs(f,b)$.  By \pref{fact:f-count}, $\calC$ visits exactly $f+1$ distinct vertices, call them $\{(v^{(1)}, i^{(1)}), \dots, (v^{(f+1)}, i^{(f+1)})\}$.  Every isomorphic $\calC'$ may be obtained by taking a list of numbers $(i'_1, \dots, i'_{f+1}) \in [n]^{f+1}$ and replacing all appearances of $(v^{(j)}, i^{(j)})$ in $\calC$ with $(v^{(j)}, i'_j)$. (Not all such lists lead to isomorphic $\calC'$, but we don't mind overcounting.)  This completes the proof, as there are $n^{f+1}$ such lists.
\end{proof}

We now have all the tools to prove the desired lemma.

\begin{proof}[Proof of \pref{lem:fb}]
    With \pref{prop:canlkgs-bound} in place, % we can worry only about bounding canonical linkages and then combine such a bound with \pref{prop:canlkgs-bound} to prove the desired lemma.
    it suffices to bound the number of canonical linkages as follows:
    $$|\Lkgs^c(f,b)| \leq \poly(k,\ell)^{b+k} \cdot \hat{\rho}^{f/2}.$$

    Our strategy is to give an encoding of linkages in $\Lkgs^c(f,b)$, and then bound the number of possible encodings.  Let $\calC$ be an arbitrary linkage in $\Lkgs^c(f,b)$.  To encode $\calC$, we first partition it into $2k$ many ``$2(\ell + 1)$-segments'', each of which corresponds to nonbacktracking walks between spurs,%\rnote{Yeah, I'm a little confused about the off-by-ones here.  Like, do you want to include the starting spur or the ending spur in each 2ell-segment?}
    and specify how to encode each $2(\ell + 1)$-segment.  We then partition each $2(\ell + 1)$-segment into maximal contiguous blocks of the same type of step (``type'' as in \pref{def:steptype}) and store an encoding of information about the steps therein.  Ultimately, it will be possible to uniquely decipher $\calC$ from its constructed encoding.

    Towards describing our encoding, we first define the sequence $S_{\visited}$, constructed from the $f+1$ vertices in $V(\calC)$ sorted in increasing order of first-visit time.

    \paragraph{Encoding positions of blocks.}  We define $P_{\fresh}$, $P_{\boundary}$ and $P_{\stale}$, which are sequences noting the starting positions and ending positions of fresh, boundary, and stale blocks respectively, in the order visited in $\calC$.

    \paragraph{Encoding fresh steps.}  Let $S_{\fresh}$ be the sequence obtained by replacing each vertex of $S_{\visited}$ with its corresponding base vertex in $K_{r,c}$.

    \paragraph{Encoding boundary steps.}  Let $\beta$ be a block of boundary steps $(v_0,v_1),\dots,(v_{|\beta|-1},v_{|\beta|})$.  Let $t_i$ be such that $v_i$ is the $t_i$-th vertex in $S_{\visited}$.  We define $\Enc_b(\beta)$ as the sequence $(t_0,t_1),\dots,(t_{|\beta|-1},t_{|\beta|})$.  Let $\beta_1,\dots,\beta_{T}$ be the blocks of boundary steps in the order in which they appear in $\calC$.  We store the concatenation of $\Enc_b(\beta_1),\dots,\Enc_b(\beta_{T})$, which we call $S_{\boundary}$.

    \paragraph{Encoding stale steps.}  For each block $\beta$ of stale steps, let $u$ be the first vertex and $v$ be the last vertex of $\beta$, and let $p({\beta})$ be the position in $\calC$ where the block $\beta$ starts.  Let $\calS_{p({\beta}),uv,|\beta|}$ denote the list (in, say, lexicographic order) of all possible nonbacktracking walks from $u$ to $v$ of length $|\beta|$ that only use edges visited by $\calC$ before position $p({\beta})$; note that $\beta$ occurs in $\calS_{p({\beta}),uv,|\beta|}$.  We let $\Enc_s(\beta) = (t, m)$ such that the $t$-th vertex in $S_{\visited}$ is the last vertex visited in $\beta$ (that is $v$), and $m$ is the position of $\beta$ in $\calS_{p({\beta}),uv,|\beta|}$.  Let $\beta_1,\dots,\beta_{T}$ be the blocks of stale steps in the order they appear in $\calC$.  We store the concatenation of $\Enc_s(\beta_1),\dots,\Enc_s(\beta_T)$, which we call $S_{\stale}$.

    \medskip

    We refer to the constructed $(P_{\fresh},P_{\boundary},P_{\stale},S_{\fresh},S_{\boundary},S_{\stale})$ as the \emph{encoding} of $\calC$.

    \paragraph{Unique reconstruction of linkage.}  In this part of the proof, we show that we can uniquely recover $\calC$ from its encoding.  First, since $\calC$ is a canonical linkage we can correctly reconstruct $S_{\visited}$ from $S_{\fresh}$ because the labels are visited in canonical (increasing) order.  From $P_{\fresh},P_{\boundary}$ and $P_{\stale}$, we can infer a partition of $[4k(\ell+1)]$ into blocks in order $\beta_1,\dots,\beta_T$ and the type of each block.  We sketch an inductive proof that shows how $\calC$ can be uniquely recovered from its encoding.  As our base case, the first block is a fresh block and hence all the steps that comprise it can be recovered from $S_{\visited}$.  Towards our inductive step, suppose we know the edges in $\calC$ from blocks $\beta_1,\dots,\beta_i$, we show how to recover the edges in $\beta_{i+1}$ from the encoding of $\calC$.  If $\beta_{i+1}$ is a fresh or boundary block, its recovery is straightforward.  Suppose $\beta_{i+1}$ is a stale block.  Then from $P_{\stale}$ and $S_{\stale}$, we can infer the last vertex $v$ visited by $\beta_{i+1}$ and the length of the block $|\beta_{i+1}|$.  We know the first vertex $u$ in $\beta_{i+1}$ and can reconstruct $\calS_{p({\beta_{i+1}}),uv,|\beta_{i+1}|}$ since we have complete information about the steps in $\calC$ prior to $\beta_{i+1}$.  We can then infer $\beta_{i+1}$ from $\calS_{p({\beta_{i+1}}),uv, |\beta_{i+1}|}$ and $S_{\stale}$.

    \paragraph{Bounding the number of metadata encodings.}  A fresh block must either be followed by a boundary step, or must occur at the end of a $2(\ell + 1)$-segment; analogously, a stale block must either be preceded by a boundary step, or must occur at the start of a $2(\ell + 1)$-segment.  Thus, the number of fresh blocks and stale blocks are each bounded by $b+2k$.  Further, the number of boundary blocks is clearly bounded by $b$.  Since there are at most $(4k(\ell + 1))^2$ distinct combinations of starting and ending positions of a block, the number of distinct possibilities that the triple $(P_{\fresh}, P_{\stale}, P_{\boundary})$ can be bounded by $(4k(\ell + 1))^{6b+8k}$.

    \paragraph{Bounding number of fresh step encodings.}
    For a fixed $P_{\fresh}$, we give an upper bound on the number of possibilities for $S_{\fresh}$.  Fixing $P_{\fresh}$ fixes a number $T$ as well as $q_1,\dots,q_T$ such that there are $T$ fresh blocks in $\calC$ and such that the $i$-th block has length $q_i$.  Let us focus on a single fresh block~$\beta$.  The sequence of vertices in $S_{\fresh}$ corresponding to $\beta$ give a nonbacktracking walk $W_{\beta}$ in the base constraint graph $K_{r,c}$.  Additionally, for a consecutive triple $(i,j,i')$ in this nonbacktracking walk, $\{i,i'\}$ must be an edge in the corresponding base instance graph $\calI_1$ due $\calC$ being an $\calA$-linkage.  Let $\wt{W}_{\beta}$ be the maximal subwalk of $W_{\beta}$ that starts and ends with a variable vertex.  Note that $\wt{W}_{\beta}$ corresponds exactly to a nomadic walk in $\calI_1$ whose length is at most $|\beta|/2$.  Now regarding $W_\beta$, either $W_{\beta}$ is equal to $\wt{W}_{\beta}$ (there is~$1$ way in which this can happen), or both the first and last steps of $W_{\beta}$ are not in $\wt{W}_{\beta}$ (there are $c^2$ ways in which this can happen), or exactly one of the first and last steps of $W_{\beta}$ is not in $\wt{W}_{\beta}$ (there are $2c$ ways in which this can happen).  This tells us that the number of distinct possibilities for $W_{\beta}$ is bounded by $(c+1)^2\delta_{\lfloor |\beta|/2\rfloor}$, where $\delta_{s}$ denotes the number of nomadic walks of length $s$ in $\calI_1$.  Thus, we obtain an upper bound of $(c+1)^{2T}\prod_{i=1}^T \delta_{\lfloor q_i/2\rfloor}$ on the number of possibilities for $S_{\fresh}$, which is bounded by $(c+1)^{2b+4k}\prod_{i=1}^T \delta_{\lfloor q_i/2\rfloor}$.  Towards simplifying the expression, we bound $\delta_s$.  Observe that for a given edge $e \in E(|\calI_1|)$, the number of nomadic walks of length $s$ starting with $e$ is given by $\|(|B_1|)^s \bm{1}_{e}\|_1$. This implies that $\delta_{s} \le \|(|B_1|)^s\|_1$, where $\|(|B_1|)^s\|_1 = \sup \{\|(|B_1|)^sx\| : \|x\|_1 = 1 \}$.

    To bound the above, first observe that we have a simple bound $\|(|B_1|)^s\|_1 \leq \kappa^s$ provided $\kappa$ is a large enough constant (for example, the maximum degree of $\calI_1$ is a possible such value).  Next, it is known that
    \[\lim_{s \to \infty} \left(\|(|B_1|)^s\|\right)^{1/s} = \rho(|B_1|),\]
    and hence  for any $\hat{\rho} > \rho(|B_1|)$, there is a constant $\ell_0$ such that $\|(|B|)^s\|_1 \leq (\hat{\rho})^s$ for all $s \geq \ell_0$. Putting these two bounds together we get that for any $s\ge\ell_0$, 
    \[
        \delta_s \leq \|(|B_1|)^s\|_1 \leq (\hat{\rho})^{s - \ell_0} \kappa^{\ell_0}.
    \]
    Thus the number of possibilities for $S_{\fresh}$ is bounded by $(c+1)^{2b + 4k} \prod_{i=1}^T (\hat{\rho})^{\lfloor q_i / 2\rfloor - \ell_0} \kappa^{\ell_0}$, which can, in turn, be bounded by $\left((c+1)^2\kappa^{\ell_0}\hat{\rho}^{-\ell_0}\right)^{b+2k} (\hat{\rho})^{f/2}$.

    \paragraph{Bounding number of stale step encodings.}
    For any stale block $\beta$, let $u$ and $v$ be the first and last visited vertices respectively.  $S_{\stale}$ specifies a number in $[f+1]$ to encode $v$, and a number between $1$ and $M$ where $M$ is the total number of nonbacktracking walks from $u$ to $v$ of length $|\beta|$.  Since the number of stale blocks is bounded by $b+2k$, the number of possibilities for what $S_{\stale}$ can be is at most $(M(f+1))^{b+2k}$.  We show that $M\le 2$, and hence translate our upper bound to $(2(f+1))^{b+2k}$.

    Since all blocks are contained within $2(\ell + 1)$-segments and the $\calA$-linkage being encoded is $4\ell$-tangle-free, the steps traversed by $\beta$ are in a connected subgraph $H$ with at most one cycle.  Our goal is to show that there are at most $2$ nonbacktracking walks of a given length $L$ between any pair of vertices $x,y$.  There is at most one nonbacktracking walk between $x$ and $y$ that does not visit vertices on $C$, the single cycle in $H$, and if such a walk exists, it is the unique shortest path.  Any nonbacktracking walk between $x$ and $y$ that visits vertices of $C$ can be broken down into 3 phases --- (i) a nonbacktracking walk from $x$ to $v_x$, the closest vertex in $C$ to $x$, (ii) a nonbacktracking walk from $v_x$ to $v_y$, the closest vertex in $C$ to $y$, (iii) a nonbacktracking walk from $v_y$ to $y$.  Phases (i) and (iii) are always of fixed length, whose sum is some $L'$.  Thus, it suffices to show that there are at most $2$ nonbacktracking walks from $v_x$ to $v_y$ of length $L-L'$.  Any nonbacktracking walk takes $r$ rotations in $C$ and then takes an acyclic path from $v_x$ to $v_y$, whose length is observed to be strictly less than $|C|$, for $r\ge 0$.  The steps in a nonbacktracking walk from $v_x$ to $v_y$ are either all in a clockwise direction, or all in an anticlockwise direction, and hence for any $r$ there are at most $2$ nonbacktracking walks from $v_x$ to $v_y$ of length strictly between $(r-1)|C|$ and $r|C|+1$.  In particular, there are at most $2$ nonbacktracking walks between $v_x$ and $v_y$ of length equal to $L-L'$.

    \paragraph{Bounding number of boundary step encodings.}  $S_{\boundary}$ is a sequence of $b$ tuples in $[f+1]^2$, and hence there are at most $(f+1)^{2b}$ distinct sequences that $S_{\boundary}$ can be.

    \paragraph{Final bound:}  The above gives us a final bound of:
    \begin{equation}
        (4k(\ell + 1))^{6b+8k} ((c+1)^2\kappa^{\ell_0}(\hat{\rho})^{-\ell_0})^{b+2k}(\hat{\rho})^{f/2}2^{b+2k}(f+1)^{3b+2k}
    \end{equation}
    which, when combined with \pref{prop:canlkgs-bound} gives the desired claim.
\end{proof}

We wrap everything up by combining the results of \pref{lem:fb} with \pref{lem:finish-trace2} to prove \pref{thm:C}.

\begin{proof}[Proof of \pref{thm:C}]
    Let $\ell = \kappa\log n$, where $\kappa$ is the  universal constant  from \pref{prop:tanglefree}, let $k$ be chosen so that $k\ell = \omega(\log n)$, let $R$ be as in \pref{lem:finish-trace2}, and let $\hat{\rho}$ be any constant greater than $\rho(|B_1|)$.  Then we have
    \begin{align*}
        R &= \sum_{\substack{(2k \times 2\ell)\textnormal{-nonbacktracking} \\ \textnormal{valid, internally $\ell$-tangle-free, singleton-free} \\ \atomlist\textnormal{-linkages } \calC \textnormal{ in } \calK_n}} n^{-|E(\calC)|}\\
        &= \sum_{f=0}^{\infty} \sum_{b=0}^{\infty} |\Lkgs(f,b)|n^{-(f+b)}\\
        &= \sum_{f=0}^{2k\ell+O(k\log(k\ell))} \sum_{b=0}^{\infty} |\Lkgs(f,b)|n^{-(f+b)} &\text{(by \pref{cor:dupe})}\\
        &\le \sum_{f=0}^{2k\ell+O(k\log(k\ell))} \sum_{b=0}^{\infty} \frac{\poly(k,\ell)^{b}\cdot\poly(k,\ell)^k\cdot (\hat{\rho})^{f/2}\cdot n}{n^b} &\text{(by \pref{lem:fb})}\\
        &= \sum_{f=0}^{2k\ell+O(k\log(k\ell))} n\cdot\poly(k,\ell)^k\cdot(\hat{\rho})^{f/2}\sum_{b=0}^{\infty}\left(\frac{\poly(k,\ell)}{n}\right)^b\\
        &= \sum_{f=0}^{2k\ell+O(k\log(k\ell))} n\cdot\poly(k,\ell)^k\cdot(\hat{\rho})^{f/2}\cdot\left(\frac{1}{1-\frac{\poly(k,\ell)}{n}}\right)\\
        &\le 2n\cdot\poly(k,\ell)^k(2k\ell + O(k\log(k\ell)))(\hat{\rho})^{k\ell + O(k\log(k\ell))}
    \end{align*}
    For the choice of $k$ and $\ell$ in the theorem statement, we can use \pref{lem:finish-trace2} to conclude that
    \[
        \rho(\bB_n) \le (1+o_n(1)) \cdot \sqrt{\hat{\rho}}.
    \]
    with probability $1-O(n^{.99})$.  Since the above bound holds for any $\hat{\rho}>\rho(|B_1|)$, for any $\eps > 0$, it can be rewritten as
    \[
        \rho(\bB_n) \le (1+o_n(1)) \cdot (1+\eps) \cdot \sqrt{\rho(|B_1|)}. \qedhere
    \]
\end{proof}

%% file: SDP-value.tex
\section{The SDP value for random two-eigenvalue CSPs} \label{sec:wrapup}

In this section, we put all the ingredients together to conclude our main theorem.  We start with an elementary and well known fact and include a short proof for self containment.

\begin{fact}    \label{fact:sdp-to-spec}
    Let $A$ be a real $n\times n$ symmetric matrix.  Then
    \begin{align*}
        \frac{1}{n}\max_{X\psdge 0, X_{ii}=1}\langle A, X\rangle &\le \lambda_{\max}(A)\\
        \frac{1}{n}\min_{X\psdge 0, X_{ii}=1}\langle A, X\rangle &\ge \lambda_{\min}(A)
    \end{align*}
\end{fact}
\begin{proof}
    We prove the upper bound below.  The proof of the lower bound is identical.
    \begin{align*}
        \frac{1}{n}\max_{X\psdge 0, X_{ii}=1}\langle A, X\rangle &\le \frac{1}{n}\max_{X\psdge 0, \tr(X)=n}\langle A, X\rangle\\
        &= \max_{X\psdge 0, \tr(X)=1}\langle A, X\rangle\\
        &= \lambda_{\max}(A).
    \end{align*}
\end{proof}

Recall $\gr:=(c-1)(-\lambda_1\lambda_2)$ and $r_X := 2\sqrt{\gr}$.

\begin{theorem}
    Let $\calA = (A_1,\dots, A_c)$ be a sequence of $r$-vertex atoms with edge weights $\pm 1$.  Let $\bcongraph_n$ denote a random $n$-lifted constraint graph and $\binstgraph_n = \atomlist(\bcongraph_n)$ an associated instance graph with $1$-wise uniform negations~$(\bxi^f_{ii'})$.  Let $\bA_n$ be the adjacency matrix of $\binstgraph_n$.  Then, with probability $1-o_n(1)$,
    \begin{align*}
        \max_{X\psdge 0, X_{ii}=1} \langle\bA_n, X\rangle &= (\lambda_1+\lambda_2+r_X\pm\eps)n\\
        \min_{X\psdge 0, X_{ii}=1} \langle\bA_n, X\rangle &= (\lambda_1+\lambda_2-r_x\pm\eps)n.
    \end{align*}
\end{theorem}
\begin{proof}
    $\max_{X\psdge 0, X_{ii}=1}\langle \bA_n, X\rangle \ge (\lambda_1+\lambda_2+r_X-\eps)n$ follows from \pref{thm:SDP-sol} and $\max_{X\psdge 0, X_{ii}=1}\langle \bA_n, X\rangle \le (\lambda_1+\lambda_2+r_X+\eps)n$ follows from \pref{fact:sdp-to-spec}.  The upper and lower bounds on $\min_{X\psdge 0, X_{ii}=1}\langle\bA_n, X\rangle$ can be determined identically.
\end{proof}